\documentclass[10pt]{article}
\usepackage{graphicx,float}
\usepackage{amsmath,amssymb,amsthm,amsfonts}
\usepackage{mathrsfs,cleveref}
\usepackage{amssymb}
\usepackage{caption}
\crefname{figure}{\textbf{Fig.}}{\textbf{Figs.}}  
\Crefname{figure}{\textbf{Fig.}}{\textbf{Figs.}} 
\usepackage{placeins}
\usepackage{booktabs}
\usepackage{psfrag,epsfig,subcaption}
\usepackage{epstopdf}
\usepackage{color}
\usepackage{subcaption}
\usepackage{cleveref}
\usepackage{authblk}

\newtheorem{definition}{Definition}[section]
\newtheorem{remark}{Remark}[section]
\newtheorem{lemma}{Lemma}[section]

\newtheorem{thm}{Theorem}[section]

\newtheorem{prop}[thm]{Proposition}

\theoremstyle{definition}
\newtheorem{defn}{Definition}[section]

\numberwithin{equation}{section}

\DeclareMathSymbol{\C}{\mathalpha}{AMSb}{"43}

\newcommand{\lam}{\lambda}

\newcommand{\bsub}{\begin{subequations}}
\newcommand{\esub}{\end{subequations}$\!$}

\begin{document}
\title{Free energy dissipation and a decomposition of general jump diffusions on $\mathbb{R}^n$ without detailed balance}

\author[1,2]{Shuyuan Fan}
\author[3]{Qi Zhang}
\affil[1]{School of Mathematics, Sun Yat-sen University,  Guangzhou, Guangdong 510275, China}
\affil[2]{School of Science, Great Bay University, Dongguan, Guangdong, 523000, China}
\affil[3]{Beijing institute of mathematical sciences and applications, Beijing, 101408 China, Email: \texttt{qzhang@bimsa.cn}}

\smallbreak \maketitle

\begin{abstract}

We analyze the thermodynamic structure of jump diffusions combining Brownian and Poisson noise, a class of stochastic dynamics relevant to non-equilibrium statistical physics.
For such nonlocal dynamics, the free energy admits a full dissipation formula that decomposes into entropy production and housekeeping heat.
A central result is a decomposition of the generator into symmetric and anti-symmetric parts with respect to the invariant measure $\rho_\mathrm{ss}$.
The symmetric sector corresponds to a reversible dynamics and yields a nonlocal Fisher information governing free-energy decay, whereas the anti-symmetric sector generates a canonical conservative flow that produces circulation but no dissipation. Several numerical examples motivated by intracellular particle transports demonstrate how this decomposition clarifies the structure of non-equilibrium stationary states in jump-driven systems.

\end{abstract}

\vskip 0.2truein


\vskip 0.2truein

\section{Introduction}

The concept of free energy lies at the heart of statistical physics, serving as the fundamental bridge between microscopic stochastic fluctuations and macroscopic thermodynamic laws. 
For a stochastic system coupled to a thermal environment, 
the generalized free energy—mathematically formulated as the relative entropy with respect to a stationary measure—transcends its classical equilibrium definition to become a universal metric quantifying the ``distance" from the steady state. 
Its temporal evolution encapsulates the essence of the Second Law: 
the free energy dissipation reflects the monotonic loss of distinguishability as the system relaxes, acting as a robust Lyapunov functional that guarantees dynamical stability \cite{NP77,DM62,JQQ04}. 
Crucially, this dissipation is not merely a mathematical decay; 
it physically encodes the time-irreversibility of the underlying trajectory and provides a rigorous quantification of the entropy production rate, thereby revealing the thermodynamic cost and the arrow of time embedded in the stochastic dynamics \cite{stochasticThermo1,stochasticThermo2,GQ10,QQT02,VE10,EV10}.

Classically, for diffusion processes governed by the second-order Fokker--Planck equation, 
the relationship between free energy and entropy is well-understood: in detailed-balanced systems, the free energy dissipation rate is exactly balanced by the total entropy production rate \cite{QQT02,VE10, DP23}. 
However, for systems driven into a non-equilibrium steady state by non-conservative forces, this equality breaks down. The discrepancy led to the concept of ``housekeeping heat'' $Q_\mathrm{hk}$, first introduced phenomenologically by Oono and Paniconi \cite{OP98} to describe the heat dissipated solely to maintain the steady state. 
This concept was later formalized by Hatano and Sasa \cite{HS01} through a generalized fluctuation-dissipation relation. 
More recently, a unified framework was established by Esposito and Van den Broeck \cite{esposito2010,EV10,VE10}, who decomposed the total entropy production into an adiabatic part (associated with housekeeping heat) and a non-adiabatic part (associated with system relaxation).
In this framework, the free energy dissipation is linked strictly to the non-adiabatic entropy production, 
while the housekeeping heat represents the energetic cost of breaking detailed balance. 

Underpinning these thermodynamic phenomenologies is the geometric structure of the generator and associated Fokker-Planck operator.
For continuous diffusion processes, it is known that the drift vector field $b(x)$ admits a Helmholtz-Hodge decomposition: 
$$b(x) = a(x)a^\mathrm{T}(x)\nabla V(x) + v_\mathrm{res}(x),$$
where $V$ is a quasi-potential function linked to the stationary density $\rho_\mathrm{ss}$, the matrix $a(\cdot)$ is the diffusion coefficient, and $v_\mathrm{res}$ is a divergence-free residual drift \cite{A04,Q13, DP23}.
The vector field $b(x)$ decomposition corresponds to an operator-theoretic splitting of the generator $\mathcal{L}$ into a symmetric part $\mathcal{L}_\mathrm{s}$ and an antisymmetric part $\mathcal{L}_\mathrm{a}$ in the weighted Hilbert space $L^2(\rho_\mathrm{ss})$ \cite{Q13,DP23}. 
The symmetric part generates a reversible dynamics responsible for the relaxation towards equilibrium, while the antisymmetric part generates a measure-preserving conservative flow along the level sets of the stationary distribution. 

While the above framework is well-established for continuous diffusions, many natural systems evolve through abrupt events rather than continuous fluctuations or discrete setting. 
In particular, jump noise has emerged as a ubiquitous modeling ingredient across scales. 
In biological physics, jump diffusions and L\'evy-type models have been used to describe intermittent search and foraging patterns \cite{Viswanathan99,Bartumeus05}, bursty intracellular transport and active fluctuations in crowded cytoplasm \cite{LVS2015}, and heavy-tailed activity statistics in neuronal populations \cite{Beggs03}. 
In complex physical environments, jump components capture anomalous transport in turbulence \cite{Shlesinger87} and rapid transitions (tipping events) in climate records \cite{Ditlevsen10}. 
In quantitative finance, jump diffusion models extend Black--Scholes to reproduce skewness and fat tails in returns \cite{Merton76,DK08}. 
More recently, time-reversal and nonlocal diffusion mechanisms associated with jump processes have found new roles in machine learning, for instance in score-based generative modeling where non-Gaussian noise can improve exploration and sampling efficiency \cite{YP23,GeneratorMatching}.

A concrete biophysical motivation comes from intracellular active transport \cite{briane2020,appert2015}. 
Single-particle tracking experiments reveal strongly jump trajectories in intracellular active transport: cargos undergo long periods of near-diffusive motion and transient trapping, punctuated by sporadic fast relocations when they engage cytoskeletal tracks and molecular motors. 
This ``pause--run--pause'' dynamics naturally produces non-Gaussian displacement statistics and heavy-tailed increments, for which compound-Poisson jump diffusions provide an effective description. 
Thermodynamically, motor activity and cellular maintenance are powered by chemical free energy (ATP/GTP hydrolysis), so even when the spatial distribution of cargos appears stationary, the system can sustain persistent probability currents and a nonzero steady entropy production rate \cite{vale1985,goychuk2014,reimann2002,goychuk2014plos}. 
This highlights a key distinction that is essential in applications: free-energy dissipation quantifies relaxation toward the stationary law after perturbations, whereas housekeeping heat (and the steady component of entropy production) quantifies the ongoing energetic throughput required to maintain a non-equilibrium steady state. 
In this sense, changes in active-transport statistics can be interpreted as changes in the non-equilibrium maintenance cost, providing a thermodynamic lens on shifts between quiescent-like and activated-like cellular regimes.

Despite the broad phenomenological relevance of jump processes, the non-equilibrium statistical mechanics of continuous-state jump diffusions is far less developed than that of either continuous Langevin dynamics or discrete-state Markov jump processes (Master equations) \cite{bremaud2020,MN03,MNW08,Schnakenberg76}. 
Recent theoretical efforts have advanced the stochastic thermodynamics of L\'{e}vy flights, exploring entropy production and fluctuation theorems within frameworks featuring non-local operators. These efforts range from modeling anomalous diffusion \cite{KRS2008} and thermodynamic formulations for master equation systems \cite{TD12}, to machine-learning-based simulations of active matter \cite{BV24} and the analysis of fluctuation theorems under L\'{e}vy noise \cite{HLMZ25}. Notably, these theoretical developments are increasingly supported by experimental evidence suggesting that L\'{e}vy-type motions may underlie intracellular active transport processes \cite{chen2015}.
However, extending the geometric and operator-theoretic decompositions to this setting poses significant analytical challenges.
Unlike the local Laplacian operator in continuous diffusions, the generator of a jump process is non-local, involving integro-differential operators.
Consequently, the standard definition of "probability current" becomes non-local, making the characterization of the "residual drift" and the associated conservative dynamics highly non-trivial, and the Helmholtz-Hodge decomposition that underpin dissipation theory in the diffusion case do not transfer directly for jump diffusions. As a result, explicit and computable formulas for free-energy dissipation, entropy production, and housekeeping heat---together with a clear structural interpretation in terms of symmetric/antisymmetric components of the generator---remain less explored for general nonlocal jump diffusions. 
Bridging this gap is important both conceptually to understand irreversibility beyond the diffusive setting, and practically to quantify non-equilibrium costs in jump-driven models used in biophysics, climate, finance, and modern generative modeling.

In this paper, we are interested in the jump diffusion $(X_t)_{t\geq 0}$ which is given by the following SDE on $\mathbb{R}^n$:
\begin{equation}\label{LeqX}
    dX_t =  b(X_{t})dt +  \sqrt{2\beta^{-1}} a(X_{t^-}) \circ dB_t  +  \int_{\mathbb{R}^n\setminus\{0\}} \sigma(X_{t^{-}},z) N(dt,dz),
\end{equation}
where $B_t$ is an $n$-dimensional Brownian motion on the probability space $(\Omega, \mathscr{F}, \mathbb{P})$ equipped with a filtration $(\mathscr{F}_t)_{t\geq 0}$, $\circ$ denotes the Stratonovich integral, $b(\cdot)$ is a drift from $\mathbb{R}^n$ to $\mathbb{R}^n$, $a(\cdot)=(a_{ij}(\cdot))_{n\times n}:\mathbb{R}^n \rightarrow  \mathbb{R}^{n \times n}$ is the diffusion coefficient, $N(dt,dz)$ be an independent $n$-dimensional Poisson random measure with Poisson intensity rate $\lambda>0$ and L\'{e}vy measure $\nu(dz)$ on $\mathbb{R}^n \setminus\{0\}$, the function $\sigma:\mathbb{R}^n \times \mathbb{R}^n \rightarrow  \mathbb{R}^{n}$ is the coefficient of jump, and the $\beta:=\theta^{-1}$ denote the inverse temperature.
We remark that in this paper, we only consider the case that isothermal dynamics coupled to a single heat bath, thus $\beta$ is a fixed constant.
Throughout this paper, the jump diffusion $(X_t)_{t\ge0}$ denotes a general continuous-state Markov stochastic dynamics with possible jumps. We do not restrict $X_t$ to be a mechanical position variable. In applications, $X_t$ may represent an overdamped Langevin dynamics driven by thermal noise plus jump noise, but it may also encode other effective observables in physics, biology, or finance, such as a reaction coordinate, a concentration/abundance variable, a membrane-potential proxy, or a log-price/volatility factor. 
We present some preliminary results for jump-diffusions $(X_t)_{t\geq 0}$, and list some necessary assumptions for the SDE \eqref{LeqX} in Section \ref{Sec2}.

We assume that the process admits an invariant probability measure $\mu_\mathrm{ss}(dx)=\rho_\mathrm{ss}(x)\,dx$, and interpret $\rho_\mathrm{ss}$ as a steady state (either an equilibrium state or a non-equilibrium steady state) of the dynamics. 
Following the standard viewpoint in stochastic thermodynamics, we define the generalized free energy of $X_t$ relative to the stationary measure $\mu_\mathrm{ss}$ by the relative entropy
\[
F(\rho(t)) 
= \beta^{-1}\mathcal{H}(\mu_t\mid\mu_\mathrm{ss})
= \beta^{-1}\int_{\mathbb{R}^n} \log\left(\frac{\rho(t,x)}{\rho_\mathrm{ss}(x)}\right)\rho(t,x)\,dx,
\]
where $\rho(t,x)$ is the density function of $X_t$. In Section~\ref{Sec3}, combining this definition with the nonlocal Fokker–Planck equation, we derive an explicit formula for the free energy dissipation rate $dF/dt$ in terms of the local and nonlocal probability currents $j^{\rm loc}, j^{\rm nl}$. 
More precisely, we obtain a decomposition of the form
\[
\frac{dF(\rho(t))}{dt}
= Q_\mathrm{hk}(t) - \beta^{-1} e_\mathrm{p}(t),
\]
where $e_\mathrm{p}(t)\ge0$ is the entropy production rate, and $Q_\mathrm{hk}(t)\ge0$ is the housekeeping heat associated with the maintenance of a non-equilibrium steady state. 
The entropy production rate $e_\mathrm{p}(t)$ contains both a local diffusive contribution and a nonlocal jump contribution, each expressed in a fluctuation–dissipation form involving suitable logarithmic ratios of forward and backward transition rates. 
We also prove that $F(\rho(t))$ is nonincreasing in time and establish a Clausius-type inequality
\[
\theta\frac{dS(t)}{dt} - \frac{dQ(t)}{dt} = \theta e_\mathrm{p}(t)\ge 0,
\]
where $S(t)$ is the Gibbs entropy and $Q(t)$ is the heat flux exchanged with the environment. 
These results extend to general jump diffusions the thermodynamic structure previously known for diffusions and Markov jump processes \cite{EV10,VE10,MN03,MNW08,ZL25}.

A second main objective of this work is to clarify the structural origin of free energy dissipation in terms of the generator of the jump diffusion. 
Motivated by the diffusion case \cite{Q13,DP23,ZL25}, we introduce the weighted Hilbert space $L^2(\mathbb{R}^n,\rho_\mathrm{ss}(x)\,dx)$ and study the adjoint structure of the generator $\mathcal{L}$ with respect to the inner product
\[
\langle \varphi,\psi\rangle_\mathrm{ss}
= \int_{\mathbb{R}^n}\varphi(x)\psi(x)\rho_\mathrm{ss}(x)\,dx.
\]
We show in Section~4 that $\mathcal{L}$ admits a canonical decomposition into a symmetric and an anti-symmetric part in this weighted $L^2$–space:
\[
\mathcal{L} = \mathcal{L}_\mathrm{s} + \mathcal{L}_\mathrm{a},
\]
where $\mathcal{L}_\mathrm{s}$ is self–adjoint with respect to $\mu_\mathrm{ss}$, while $\mathcal{L}_\mathrm{a}$ is anti-symmetric with respect to $\mu_\mathrm{ss}$. 

Here, the anti-symmetric operator $\mathcal{L}_\mathrm{a}$ encapsulates the "rotational" component of the stochastic dynamics. 
In the absence of detailed balance, the stationary condition $\mathcal{L}^*\rho_\mathrm{ss}=0$ is satisfied not locally by vanishing currents, but globally by the formation of closed current loops.
The operator $\mathcal{L}_\mathrm{a}$ generates this divergence-free (solenoidal) flow of probability, which we term the canonical conservative dynamics, so that the  free energy dissipation due to $\mathcal{L}_\mathrm{a}$ vanishes identically:
\[
F'_\mathrm{a}(\rho(t)) = 0.
\]
In contrast, the symmetric part $\mathcal{L}_\mathrm{s}$ has a natural Dirichlet form associated with it and defines a nonlocal Fisher information functional $\mathcal{I}(\rho(t))$, so that the free energy dissipation due to $\mathcal{L}_\mathrm{s}$ is given by
\[
\frac{dF(\rho(t))}{dt}
= F'_\mathrm{s}(\rho(t))
= -\beta^{-1}\mathcal{I}(\rho(t)),
\]
so that the housekeeping heat is exactly the gap between the total entropy production and the symmetric Fisher-information dissipation. 
This provides a new interpretation of non-equilibrium housekeeping heat in terms of a difference between two equilibrium-like reversible dynamics, both formally governed by symmetric generators.

The relationship between entropy production, time–reversibility, and detailed balance for general jump diffusions has been recently clarified in \cite{ZL25}. 
Under suitable assumptions, detailed balance, reversibility with respect to $\mu_\mathrm{ss}$, a certain gradient structure of the drift and jump kernel, and the vanishing of the stationary entropy production rate are all equivalent (see Proposition~\ref{TSeq}). 
Our analysis refines this picture by identifying, in the absence of detailed balance, a canonical decomposition of the dynamics into a non-dissipative conservative part and a purely dissipative symmetric part, and by linking the free energy dissipation to a nonlocal Fisher information associated with the symmetric generator.

To illustrate our results, we apply the free energy dissipation and the decomposition structure of jump-diffusions to analyze the intracellular active transports in Section~5. 
The first example in Section 5 is a L\'evy-driven jump-diffusion motivated by intermittent intracellular active transport, exhibiting a non-equilibrium steady state where the entropy production rate and housekeeping heat have nontrivial steady values and the free energy decays to zero while the system remains out of equilibrium. The second example is a intracellular active transport with an equilibrium steady state (Gibbs measure), for which the housekeeping heat vanishes, and the free energy dissipation is entirely given by the entropy production rate. 
These examples demonstrate how our decomposition clarifies the respective roles of symmetric and anti-symmetric dynamics in the thermodynamic behavior of jump diffusions.

In summary, our contributions are threefold: 
(i) we derive an explicit free-energy dissipation identity for general jump diffusions in terms of local and nonlocal probability currents, yielding a transparent decomposition into entropy production and housekeeping heat; 
(ii) we establish a canonical weighted-$L^2(\rho_{\mathrm{ss}})$ symmetric/antisymmetric decomposition of the nonlocal generator and identify the symmetric dissipation with a nonlocal Fisher-information/Dirichlet-form structure; 
and (iii) we illustrate how these objects behave in representative examples, including a prototypical non-equilibrium L\'evy-driven model that captures heavy-tailed, bursty transport relevant to biophysical applications.

The rest of the paper is organized as follows. 
In Section \ref{Sec2}, we recall basic facts on jump diffusions and nonlocal Fokker–Planck equations, and we collect the assumptions ensuring existence of a unique strong solution, a smooth transition density, and an invariant measure. 
In Section~\ref{Sec3}, we introduce the generalized free energy, derive its dissipation formula, and obtain the decomposition into entropy production and housekeeping heat together with the Clausius inequality. 
Section~4 is devoted to the $L^2(\rho_\mathrm{ss})$–decomposition of the generator into symmetric and anti-symmetric parts, the notion of canonical conservative dynamics, and the characterization of free energy dissipation in terms of the symmetric Dirichlet form and Fisher information. In Section~5, we present two explicit one-dimensional examples  from intracellular active transports to illustrate the theoretical results.

\section{Preliminary}\label{Sec2}

In this section, we introduce our assumptions, and recall some basic definitions and useful results for the jump-diffusions $(X_t)_{t \geq 0}$. 
These materials can be found in some classical references for jump-diffusions or L\'{e}vy processes, including \cite{A2004, S99, K19} and the references therein.

Let $(\Omega, \mathscr{F}, (\mathscr{F}_t)_{t\geq 0}, \mathbb{P})$ be a filtered probability space that satisfies the usual hypotheses.
Let $B_t$ be an $n$-dimensional Brownian motion, and $L_{t}$ be an independent $n$-dimensional compound Poisson process with Poisson intensity rate $\lambda>0$ and the L\'{e}vy measure (probability of jump size) $\nu(dz)$ on $\mathbb{R}^n \setminus\{0\}$. 
Then the compound Poisson process is given by $L_t = \sum_{i=1}^{n_t}D_i$, where $n_t$ is the counting variable of a Poisson process with rate $\lambda$, and $\{D_i\}_{i\geq 1}$ are i.i.d. random variables with L\'{e}vy measure $\nu(dz)$, 
which are independent of $n_t$. Then we can introduce the Poisson random measure $N(dt,dy)$ associated with $L_t$ as
\begin{equation*}
    N(t,B)(\omega):= \sharp \{s\in [0,t]:\Delta L_{s}(\omega) \in B\}, \quad t\geq 0, B\in \mathscr{B}(\mathbb{R}^{n}\backslash \{0\}),
\end{equation*}
where $B \in \mathbb{R}^{n}\backslash \{0\}$ is the Borel measurable set and $\mathscr{B}(\mathbb{R}^{n}\backslash \{0\})$ is the Borel-$\sigma$-algebra on $\mathbb{R}^{n}\backslash \{0\}$, so that $\mathbb{E}N(dt,dz) = \lambda dt\nu(dz) = \lambda m(z)dtdz$. Using the Poisson random measure $N(dt,dy)$, the compound Poisson process $L_t$ can be rewritten as
\begin{equation*}
    L_{t} = \int_{0}^t\int_{\mathbb{R}^n \setminus \{0\}} z N(dt,dz).
\end{equation*}
We assume that the associated L\'{e}vy measure $\nu(dz)=m(z)dz$ satisfies the following integrability condition
\begin{equation*}
    \int_{\mathbb{R}^n \setminus\{0\}} 1\wedge |z|^2 \nu(dz) = \int_{\mathbb{R}^n \setminus\{0\}} 1\wedge |z|^2 m(z)dz < \infty.
\end{equation*}
The above integrability condition implies that the sample paths of $L_t$ have a finite quadratic variation with finitely many large jumps with an amplitude larger than $1$, almost surely.

Using the compensated Poisson random measure $\widetilde{N}(dt,dy)=N(dt,dy)- \lambda dt\nu(dy)$ and the properties of Stratonovich integral \cite{e2019}, the SDE \eqref{LeqX} can be rewritten as
\begin{align}
 dX_t = & \left[ b(X_{t}) + \beta^{-1}\nabla \cdot A(X_t) + \lambda\int_{0< |z| \leq 1} \sigma(X_{t^-},z) \nu(dz) \right]dt + \sqrt{2\beta^{-1}} a(X_{t^-})dB_t \nonumber \\
 & \quad + \int_{0<|z| \leq 1} \sigma(X_{t^{-}},z)\widetilde{N}(dt,dz) + \int_{|z|>1} \sigma(X_{t^{-}},z) N(dt,dz),
\end{align}
where $A(x) = a(x)a^\mathrm{T}(x):= (A_{ij}(x))$ and $(\nabla\cdot A(x))_i:=\sum_{jk}a_{kj}\partial_ka_{ij}$.

Now we make some assumptions for the SDE \eqref{LeqX}. All these assumptions in this section are assumed to hold in the sequel unless otherwise specified.
\begin{enumerate}
    \item[({\bf A})](Lipschitz condition) For all $k \in \mathbb{N}$ and $ x_{1}, x_{2} \in  \mathbb{R}^n$ with $|x_1| \vee |x_2| \leq k$, there exists a constant  $c_1 >0$ such that
\begin{equation*}
     |b(x_1) - b(x_2)|^2+\|A(x_1)- A(x_2)\| + \int_{0 <|z| <1} |\sigma(x_1,z) - \sigma(x_2,z)|^2 \nu (dz) \leq c_1 |x_1-x_2|^{2}.
\end{equation*}
Here and below, $\|\cdot\|$ denotes the Hilbert-Schmidt norm of a matrix, and $|\cdot|$ denotes the Euclidean norm.
    \item[({\bf B})] (Growth condition) For all $k \in \mathbb{N}$ and $ x\in  \mathbb{R}^n$, there exists a constant  $c_2 >0$ such that
\begin{equation*}
 |b(x)|^2 + \|A(x)\| + \int_{0 <|z| <1} |\sigma(x,z)|^2 \nu (dz) < c_2(1 + |x|^2).
\end{equation*}
    \item[({\bf C})] (Boundedness condition) the Jacobian matrix $\nabla_z \sigma(x,z)$ is non-degenerate for all $x,z \in \mathbb{R}^n$, and there exists a constant $c_2>0$ such that $\|\nabla_z \sigma(x,z)\| < c_2$ for all $x,z \in \mathbb{R}^n$.
\end{enumerate} 

Under above assumption, we have the following well-posedness result of the SDE \eqref{LeqX} (see e.g. \cite[Theorem 6.2.9]{A2004} or \cite[Theorem 3.3.1]{K19}).

\begin{thm}\label{ExD}
Suppose that $({\bf A})$ and $({\bf B})$ hold, and impose the standard initial condition. Then the SDE \eqref{LeqX} admits a unique global strong solution $(X_t)_{t\geq 0}$.
\end{thm}

The solution $(X_t)_{t\geq 0}$ has the Markov semigroup $(P_t)_{t\geq 0}$, which is given by
\begin{equation*}
    P_t f(x)=\mathbb{E}(f(X_t)|X_0 =x), \quad \forall f \in L^1(\mathbb{R}^n, \mu).
\end{equation*}
We also denote the generator of the Markov process $(X_t)_{t\geq 0}$ by $\mathcal{L}$, i.e.
\begin{equation*}
    \mathcal{L}f = \lim_{t \rightarrow 0} \frac{1}{t}(P_t f - f),\quad \forall f \in L^1(\mathbb{R}^n, \mu).
\end{equation*}
Then the jump diffusion $X_t$ has generator
\begin{align*}
    \mathcal{L}f(x) = & b^\mathrm{T}(x) \nabla f(x) + \beta^{-1} \nabla\cdot (A(x)\nabla f(x)) + \lambda\int_{\mathbb{R}^n \setminus \{ 0 \} } \left[ f(x+ \sigma(x,z)) - f(x)\right] d\nu(z).
\end{align*}

Under the above boundedness condition, for every $x \in \mathbb{R}^n$, the map $z \mapsto \sigma(x,z)$ is a $C^1$-diffeomorphism, and admits an inverse denoted by $\sigma^{-1}(x,z)$. By the change of variables, the generator can also be rewritten as
\begin{align*}
    \mathcal{L}f(x) = & b^\mathrm{T}(x) \nabla f(x) + \beta^{-1} \nabla\cdot (A(x)\nabla f(x))  \\
      & + \lambda \int_{\mathbb{R}^{n}\setminus \{0\}} [f(x+z)- f(x)] \det(\nabla_{z}\sigma^{-1}(x,z))m(\sigma^{-1}(x,z))dz \\
    = & b^\mathrm{T}(x) \nabla f(x) + \beta^{-1} \nabla\cdot (A(x)\nabla f(x)) +   \int_{\mathbb{R}^{n}\setminus \{x\}} [f(y)-f(x)] k(x,y)dy,
    \end{align*}
where $k(x,y) = \lambda m(\sigma^{-1}(x,y- x))\det(\nabla_{y-x}\sigma^{-1}(x,y-x))$ is the jump kernel of $X_t$. Note that $k(x,y)dtdy$ is a measure on $\mathbb{R}^{+}\times(\mathbb{R}^n\setminus\{x\})$ associated with the pure jump part $\Delta X_t = X_t - X_{t^{-}}$ of $(X_t)_{t\geq0}$. Now we introduce the Poisson random measure 
\begin{align*}
    N_X(t,B)(\omega):= \sharp \{s\in [0,t]:\Delta X_{s}(\omega) \in B\}, \quad t\geq 0, B\in \mathscr{B}(\mathbb{R}^{n}\backslash \{0\}),
\end{align*}
so that $N_X(dt,dz):= \det(\nabla_z \sigma^{-1}(X_{t^{-}},z))N(dt, d\sigma^{-1}(X_{t^{-}},z))$, where $B \in \mathbb{R}^{n}\backslash \{0\}$ is the Borel measurable set and $\mathscr{B}(\mathbb{R}^{n}\backslash \{0\})$ is the Borel-$\sigma$-algebra on $\mathbb{R}^{n}\backslash \{0\}$. Then the associated compensated Poisson random measure 
\begin{align*}
    \widetilde{N}_X(dt,dz) = & N_X(dt,dz) - \lam\det(\nabla_z \sigma^{-1}(X_{t^{-}}, z))m(\sigma^{-1}(X_{t^{-}},z))dzdt \\
    = & N_X(dt,dz) -k(X_{t^{-}},X_{t^{-}} +z)dzdt
\end{align*}
generates a local martingale.
Using the Poisson random measure $N_X(dt,dz)$, we can rewrite
\begin{equation*}
    \int_{|z|>0} \sigma(X_{t^-},z) N(dt,dz) \overset{\text d}{=} \int_{|z|>0} z N_X(dt,dz),
\end{equation*}
and the SDE \eqref{LeqX} is also equivalent in distribution sense with
\begin{equation*}
     dX_t =  b(X_{t})dt +  \sqrt{2\beta^{-1}} a(X_{t^-}) \circ dB_t  + \int_{|z|>0} z N_X(dt,dz).
\end{equation*}

The existence of a smooth transition density of Markov processes with jumps has been investigated by Malliavin calculus, see e.g. Picard \cite{P96}, and Kunita \cite[Chapter 6]{K19}.
Moreover, the probability density function $\rho(x,t)$ satisfies the nonlocal Fokker-Planck equation
\begin{equation}\label{nonlocalFk}
    \partial_t \rho(x,t) = \mathcal{L}^{\ast}\rho(x,t),
\end{equation}
where $\mathcal{L}^{\ast}$ is the adjoint operator of the generator $\mathcal{L}$, which is given by
\begin{equation*}
    (\mathcal{L}^{\ast}\rho)(x) = \nabla\cdot [-b(x)\rho(x) + \beta^{-1} A(x)\nabla \rho(x)] +\int_{\mathbb{R}^n \setminus \{ x \} }k(y,x)\rho(y) - k(x,y)\rho(x)dy.
\end{equation*}

Similar with the local Fokker-Planck equation, we define the local probability current
\begin{equation}
    j^\mathrm{loc}(t,x) := b(x)\rho(t,x) - \beta^{-1} A(x)\cdot\nabla \rho(t,x).
\end{equation}
Since there is a nonlocal jump part in the operator $\mathcal{L}^{\ast}$, for every $x \neq y$, we also define the nonlocal probability current between $x$ and $y$ as
\begin{equation}
    j^\mathrm{nl}(t,x,y) = \rho(t,x)k(x,y) - \rho(t,y)k(y,x).
\end{equation}
By the local probability current $j^\mathrm{loc}(t,x)$ and the nonlocal probability current $j^\mathrm{nl}(t,x,y)$, we can rewrite the Fokker-Planck equation \eqref{nonlocalFk} as
\begin{equation*}
    \partial_t \rho(x,t) + \nabla\cdot j^\mathrm{loc}(t,x) + \int_{\mathbb{R}^n \setminus \{x\}} j^\mathrm{nl}(t,x,y) dy=0.
\end{equation*}
By the local probability current $j^\mathrm{loc}(t,x)$ and the nonlocal probability current $j^\mathrm{nl}(t,x,y)$, we introduce the definition of detailed balance for the jump diffusion.

In physics, the process $(X_t)_{t\geq 0}$ is stationary with its invariant measure $\mu(dx)$ means that the corresponding dynamic is in a steady state (including the equilibrium state and non-equilibrium steady state).
A probability measure $\mu$ is said to be an invariant measure of Markov semigroup $P_t$ if for every $f \in \mathcal{B}_b(\mathscr{B}(\mathbb{R}^n))$,
\begin{equation*}
    \int_{\mathbb{R}^n} P_t f(x) \mu(dx) = \int_{\mathbb{R}^n} f(x) \mu(dx).
\end{equation*}
Furthermore, if the invariant measure $\mu$ has a $C^2$ density $\rho_\mathrm{ss}$, then $\rho_\mathrm{ss}$ satisfies the stationary Fokker-Planck equation $\mathcal{L}^{\ast}\rho_\mathrm{ss} = 0$. 
In order to guarantee the existence of invariant measures (steady states) of the jump diffusions $(X_t)_{t\geq 0}$, we also need the following assumption for the SDE \eqref{LeqX}.
\begin{enumerate}
    \item[({\bf D})] (Lyapunov condition) There exists a positive function $w_1 \in C^2(\mathbb{R}^n)$, and a positive compact function $w_2 \geq c_6|x|^p$ for some constant $c_6>0$ and $p \geq 1$, so that $\mathcal{L}w_1 \leq c_5 - w_2$ holds for some constant $c_5>0$.
\end{enumerate}
Under Lyapunov condition ({\bf D}), the solution has the uniform moment estimate $$\sup_{t\geq0}\mathbb{E}|X_t|^p < c^{-1}_{6} \sup_{t\geq0}\mathbb{E}w_2(X_t) < \infty.$$ 
The associated semigroup of SDE \eqref{LeqX} has an invariant probability measure $\mu$, see e.g. \cite{H80}.

Let $(X_t)_{t\ge 0}$ be the jump diffusion with generator $\mathcal L$ and assume it admits an invariant measure $\mu_{\mathrm{ss}}(dx)=\rho_{\mathrm{ss}}(x)\,dx.$

\begin{defn}[Detailed balance]
We say that $(X_t)_{t\ge0}$ satisfies the detailed balance condition (with respect to the invariant measure $\mu_{\mathrm{ss}}=\rho_{\mathrm{ss}}(x)\,dx$) if the local and nonlocal steady currents vanish, i.e.,
\[
j^{\mathrm{loc}}_{\mathrm{ss}}(x)=0 \quad \text{for all }x\in\mathbb R^n,
\qquad 
j^{\mathrm{nl}}_{\mathrm{ss}}(x,y)=0 \quad \text{for all }x\neq y,
\]
where the local steady current is defined by
\[
j^{\mathrm{loc}}_{\mathrm{ss}}(x)
:= b(x)\rho_{\mathrm{ss}}(x)-\beta^{-1} A(x)\nabla\cdot\rho_{\mathrm{ss}}(x),
\]
and the nonlocal steady current is given by
\[
j^{\mathrm{nl}}_{\mathrm{ss}}(x,y)
:= \rho_{\mathrm{ss}}(x)\,k(x,y)-\rho_{\mathrm{ss}}(y)\,k(y,x).
\]
\end{defn}

In thermodynamics theory, the process is in an equilibrium state if the stationary process (with respect to the invariant measure) satisfies the detailed balance condition. The equilibrium systems exhibit a fundamental temporal symmetry with respect to the time $t$. This symmetry is mathematically characterized by time-reversibility, where the stochastic dynamics remain statistically invariant under time reversal. 
Let $(X_t)_{t\geq 0}$ be a stationary jump diffusion satisfying the SDE \eqref{LeqX} with invariant measure $\mu(dx)$. 
The invariant measure $\mu$ is called a reversible (probability) measure of $(X_t)_{t\geq 0}$ if for every $f,g \in C_{b}(\mathbb{R}^n)$,
\begin{equation*}
    \mathbb{E}^{\mu}[f(X_t)g(X_0)]  = \mathbb{E}^{\mu}[f(X_0)g(X_t)].
\end{equation*}

The connection between self-adjointness of the generator of an ergodic Markov process in $L^2(\mu)$ and the time-reversibility Markov process is established by Fukushima and Stroock \cite[Theorem 2.3]{FS86}.This result show that a invariant measure $\mu$ is reversible with respect to $(X_t)_{t\geq 0}$ if and only if the generator $\mathcal{L}$ is symmetric in the domain of $\mathcal{L}$: $D(\mathcal{L}) \cap L^2(\mathbb{R}^n, \mu)$, i.e. for every $f,g \in D(\mathcal{L}) \cap L^2(\mathbb{R}^n, \mu)$,
\begin{equation*}
    \int_{\mathbb{R}^n} f(x) \mathcal{L}g(x) \mu(dx) = \int_{\mathbb{R}^n} g(x) \mathcal{L} f(x) \mu(dx).
\end{equation*}
We also assume that the process $(X_t)_{t \geq 0}$ satisfies the following condition to ensure that the well-poseedness of the entropy production.
\begin{enumerate}
    \item[({\bf E})] (Regularity condition) For every $t \geq 0$, we assume that $\nabla\log \rho(t,\cdot)$ is $\mu^{X_t}-$a.s. bounded, where $\mu^{X_t}$ is the distribution of $X_t$. We also suppose that the jump kernel $k(x,y)$ satisfies that for every $t>0$,
\begin{equation}\label{ratiolim}
    \bar{k}(x,y):= k(x,y)\log \frac{k(x,y)}{k(y,x)} \in L^{\infty}(\mu^{X_t}(dx)dy),
\end{equation}
\end{enumerate} 

As established in \cite[Theorem 5.1]{ZL25}, we have the following theorem characterizing the equivalence among detailed balance, time-reversibility, along with explicit drift-kernel structures and $0$ entropy production rate.

\begin{prop}\label{TSeq}
Suppose that the assumptions (A)(B)(C)(D)(E) hold.
Let $\mu(dx) = \rho_\mathrm{ss}(x)dx$ be the Gibbs measure.
Then the following are equivalent: 
\begin{enumerate}
    \item The stationary process $(X_t)_{t\geq 0}$ with invariant density $\rho_\mathrm{ss}$ is time-reversible.
    \item System $X_t$ with its invariant measure $\mu(dx)=\rho_\mathrm{ss}(x)dx$ obeys detailed balance.
    \item The drift $b(x) = -A(x)\nabla V(x)$, and the jump kernel have a discrete gradient structure
    \begin{equation*}
         k(x,y) = \mathrm{e}^{-\beta[V(y)-V(x)]/2} s(x,y),
    \end{equation*}
    where $s(x,y) = s(y,x)$ is a nonnegative symmetric function from $\mathbb{R}^n \times \mathbb{R}^n$ to $\mathbb{R}$.
    \item The entropy production rate $\mathrm{e}^\mathrm{ss}_\mathrm{p} = 0$ at the stationary measure $\mu(dx) :=\rho_\mathrm{ss}(x)dx$, where the entropy production rate is defined as \eqref{Def:EPR}.
\end{enumerate}
\end{prop}

\section{Free energy dissipation}\label{Sec3}

In this section, we study the thermodynamics for jump diffusions based on relative entropy. 
Suppose that the stationary measure $\mu_\mathrm{ss}$ has density $\rho_\mathrm{ss}(x) = Z^{-1}\mathrm{e}^{-\beta V(x)}$, where $V$ is the quasi-potential and $Z=\int_{\mathbb{R}^n}\mathrm{e}^{-\beta V}dx$ is normalization constant. We remark that under detailed balance condition, the drift $b(x)$ and the jump kernel $k(x,y)$ are determined by the potential $V(x)$, 
and the stationary measure $\mu_\mathrm{ss}$ is given by the Gibbs measure: $\mu(dx) = Z^{-1} \mathrm{e}^{-\beta V(x)}dx$.
Thus the quasi-potential $V(x)$ can be viewed as a natural extension of the classical quasi-potential energy in thermodynamics.

Since the diffusion $X_t$ is moving in the quasi-potential $V(x)$, the total internal energy of $X_t$ is given by
\begin{equation*}
    U(t):= \int_{\mathbb{R}^n} V(x) \rho(t,x) dx = \mathbb{E}V(X_t).
\end{equation*}
Moreover, the Gibbs entropy of the jump diffusion $(X_t)_{t \geq 0}$ is 
\begin{equation*}
    S(t) = -\int_{\mathbb{R}^n}\rho(t,x)\log \rho(t,x) dx,
\end{equation*}
where $\rho(t,\cdot)$ is the probability density of $X_t$.
In the standard isothermal heat-bath setting, the free energy of $X_t$ is defined as 
\begin{equation}\label{freeener:DB}
    F(\rho(t)) := U(t) - \theta S(t) = \int_{\mathbb{R}^n} [V(x) - \theta \log \rho(t,x)]\rho(t,x) dx = \beta^{-1}\mathcal{H}(\mu_{t} | \mu_\mathrm{ss}) .
\end{equation}
where $\theta = \beta^{-1}$ is the temperature of the heat bath, and $\mathcal{H}(\mu_{t} | \mu_\mathrm{ss})$ is the relative entropy between $\mu_{t}$ and the Gibbs measure $\mu_\mathrm{ss}(dx) =  Z^{-1} \mathrm{e}^{-\beta V(x)}dx$.
Under detailed balance with respect to an invariant density $\rho_{\mathrm{ss}}$, the stationary state $\mu_{\mathrm{ss}}(dx)=\rho_{\mathrm{ss}}(x)\,dx$ is an equilibrium state: steady currents vanish and the steady entropy production rate is zero. Then the First Law implies that changes of $U(t)$ arise from heat exchanged with the bath, while the generalized free energy is nonincreasing and relaxes to its minimum value $0$ at equilibrium (see \cite{ZL25}).

Without the detailed balance condition, motivated by the detailed balance case, one can define the generalized free energy of $X_t$ by the relative entropy between $\mu_{t}$ and $\mu_\mathrm{ss}$:
\begin{equation*}
    F(\rho(t)) = \beta^{-1}\mathcal{H}(\mu_{t} | \mu_\mathrm{ss}) = \beta^{-1}\int_{\mathbb{R}^n} \log \left( \frac{\rho(t,x)}{\rho_\mathrm{ss}(x)} \right)\rho(t,x)dx.
\end{equation*}
Since $\log r \leq r-1$, $r\in \mathbb{R}$, we have
\begin{equation*}
    F(\rho(t)) =   -\beta^{-1}\int_{\mathbb{R}^n} \log \left( \frac{\rho_\mathrm{ss}(x)}{\rho(t,x)} \right)\rho(t,x)dx \geq -\beta^{-1}\int [\rho_\mathrm{ss}(x) - \rho(t,x)]dx = 0.
\end{equation*}
The equality holds if and only if $\rho_\mathrm{ss}(x) = \rho(t,x)$. Combining with the nonlocal Fokker-Planck equation, the free energy dissipation is given by
\begin{align}\label{FEdiss}
   \frac{dF(\rho(t))}{dt}
   = & \beta^{-1}\frac{d}{dt} \int_{\mathbb{R}^n} \log \left( \frac{\rho(t,x)}{\rho_\mathrm{ss}(x)}\right)\rho(t,x)dx \nonumber\\
   = & \beta^{-1}\int_{\mathbb{R}^n} \partial_t \rho(t,x) \left(\log \left( \frac{\rho(t,x)}{\rho_\mathrm{ss}(x)} \right) +1\right) dx \nonumber\\
   = & -\beta^{-1}\int_{\mathbb{R}^n} [ -b(x)\rho(t,x) + \beta^{-1} A(x)\nabla \rho(t,x)]^\mathrm{T}\nabla \log \left( \frac{\rho(t,x)}{\rho_\mathrm{ss}(x)} \right)dx \nonumber \\
   & - \beta^{-1}\int_{\mathbb{R}^n}\int_{\mathbb{R}^n \setminus \{x\}} [\rho(t,x)k(x,y) - \rho(t,y)k(y,x)]\log \left( \frac{\rho(t,x)}{\rho_\mathrm{ss}(x)} \right) dydx \nonumber\\
   = & -\frac{1}{\beta}\int_{\mathbb{R}^n} [ b(x)\rho(t,x) - \beta^{-1} A(x)\nabla \rho(t,x)]^\mathrm{T}(\nabla \log\rho_\mathrm{ss}(x) - \nabla \log\rho(t,x)) dx \nonumber \\
   & - \frac{1}{2\beta}\int_{\mathbb{R}^n}\int_{\mathbb{R}^n \setminus \{x\}} [\rho(t,x)k(x,y) - \rho(t,y)k(y,x)]\log \left( \frac{\rho(t,x) \rho_\mathrm{ss}(y)}{\rho(t,y)\rho_\mathrm{ss}(x)} \right) dydx \nonumber\\
    = & -\frac{1}{\beta}\int_{\mathbb{R}^n} (j^\mathrm{loc}(t,x))^\mathrm{T}(\nabla \log\rho_\mathrm{ss}(x) - \nabla \log\rho(t,x)) dx \nonumber \\
    & - \frac{1}{2\beta}\int_{\mathbb{R}^n}\int_{\mathbb{R}^n \setminus \{x\}} j^\mathrm{nl}(t,x,y)\log \left( \frac{\rho(t,x) \rho_\mathrm{ss}(y)}{\rho(t,y)\rho_\mathrm{ss}(x)} \right) dydx.
\end{align}
We decompose
\begin{equation*}
    \frac{dF(\rho(t))}{dt} = Q_\mathrm{hk}(t)- \beta^{-1}e_\mathrm{p}(t),
\end{equation*}
where 
\begin{align}\label{Def:Qhk}
    Q_\mathrm{hk}(t):= & \frac{1}{\beta}\int_{\mathbb{R}^n} (b(x)\rho(t,x) - \beta^{-1} A(x)\nabla \rho(t,x))^\mathrm{T}(\beta A^{-1}(x)b(x) - \nabla \log\rho_\mathrm{ss}(x) ) dx \nonumber\\
    & + \frac{1}{2\beta}\int_{\mathbb{R}^n}\int_{\mathbb{R}^n \setminus \{x\}} [\rho(t,x)k(x,y) - \rho(t,y)k(y,x)]\log \left( \frac{\rho_\mathrm{ss}(x) k(x,y)}{\rho_\mathrm{ss}(y)k(y,x)} \right) dydx 
\end{align}
is the housekeeping heat, and in analogy to the classical case \cite{S05} the entropy production rate is given by
\begin{align}\label{Def:EPR}
    e_\mathrm{p}(t) := & \beta\int_{\mathbb{R}^n}  \frac{(j^\mathrm{loc}(t,x))^\mathrm{T}A^{-1}(x)j^\mathrm{loc}(t,x)}{\rho(t,x)} dx \nonumber\\
    & + \frac{1}{2} \int_{\mathbb{R}^n}\int_{\mathbb{R}^n \setminus \{x\}} [\rho(t,x)k(x,y) - \rho(t,y)k(y,x)]\log \left( \frac{\rho(t,x)k(x,y)}{\rho(t,y)k(y,x)} \right) dydx.
\end{align}
We remark that the condition \eqref{ratiolim} is proposed to ensure that the entropy production rate is well-defined.

Note that under the detailed balance condition, the dissipation rate of free energy \eqref{freeener:DB} can be determined by the variation of Gibbs entropy and the change of internal energy. {\color{blue}The variation of Gibbs} entropy in the system is
\begin{align}\label{Dentropy}
   \frac{dS(t)}{dt}  = & -\frac{d}{dt} \int_{\mathbb{R}^n}\rho(t,x)\log \rho(t,x) dx \nonumber \\
                  = & -\int_{\mathbb{R}^n} j^\mathrm{loc}(t,x)\nabla \log \rho(t,x)dx  +\frac{1}{2} \int_{\mathbb{R}^n}\int_{\mathbb{R}^n \setminus \{x\}} j^\mathrm{nl}(t,x,y)\log \left( \frac{\rho(t,x)}{\rho(t,y)} \right) dydx.
\end{align}
The change of internal energy is 
\begin{align}\label{Din}
    \frac{d U(t)}{dt}
    = & \int_{\mathbb{R}^n} V(x) \mathcal{L}^{\ast} \rho(t,x)dx \nonumber\\
    = & \int_{\mathbb{R}^n} (\nabla V(x))^\mathrm{T} j^\mathrm{loc}(t,x)dx - \frac{1}{2}\int_{\mathbb{R}^n} \int_{\mathbb{R}^n \setminus \{x\}} j^\mathrm{nl}(t,x,y)  [V(x) - V(y)] dydx.
\end{align}
Therefore, combining \eqref{Dentropy} and \eqref{Din}, the free energy dissipation under detailed balance condition is given by the entropy production rate alone:
\begin{align*}
    \frac{dF(\rho(t))}{dt}
    = & \frac{dU(t)}{dt} - \theta \frac{dS(t)}{dt} \\
    = & -\int_{\mathbb{R}^n} \frac{(j^\mathrm{loc}(t,x))^\mathrm{T}A^{-1}(x)j^\mathrm{loc}(t,x)}{\rho(t,x)} dx \\
    & - \frac{1}{2} \int_{\mathbb{R}^n} \int_{\mathbb{R}^n \setminus \{x\}} j^\mathrm{nl}(t,x,y)\log \left( \frac{\rho(t,x)k(x,y)}{\rho(t,y)k(y,x)} \right) dydx\\
    = & - e_\mathrm{p}(t) \leq 0.
\end{align*}
Next, we give the physical interpretation of the housekeeping heat $Q_\mathrm{hk}(t)$.
Without the detailed balance condition, taking the quasi-potential $V(x) = -\beta^{-1}\log Z\rho_\mathrm{ss}(x)$, the vector field can be decomposed as 
\begin{equation*}
    A^{-1}(x)b(x) = -\nabla V(x) + A^{-1}(x)b(x) + \nabla V(x) =: - \nabla V(x) + f^\mathrm{loc}(x).
\end{equation*}
where $-\nabla V$ is the internal force from the quasi-potential $V$, and the divergence-free vector field $\beta f^\mathrm{loc}=\beta A^{-1}(x)b(x) - \nabla\log\rho_\mathrm{ss}(x)$ is the thermodynamic force from outside the system. This decomposition for diffusion processes was well studied by \cite{A04, Q13}.
Moreover, motivated by the case of the discrete jump processes in \cite{MNW08}, we also decompose the jump kernel as
\begin{align*}
    \frac{k(x,y)}{k(y,x)}
    &= \frac{\rho_\mathrm{ss}(x)k(x,y)}{\rho_\mathrm{ss}(y)k(y,x)} \frac{\rho_\mathrm{ss}(y)}{\rho_\mathrm{ss}(x)}\\
    &= \exp\left( \log \frac{\rho_\mathrm{ss}(x)k(x,y)}{\rho_\mathrm{ss}(y)k(y,x)} + \log \frac{\rho_\mathrm{ss}(y)}{\rho_\mathrm{ss}(x)} \right)\\
    &= \exp\beta\left( \beta^{-1}\log \frac{\rho_\mathrm{ss}(x)k(x,y)}{\rho_\mathrm{ss}(y)k(y,x)} + V(x)-V(y) \right)\\
    :&= \mathrm{e}^{\beta(f^\mathrm{nl}(x,y) + V(x)-V(y))}.
\end{align*}
where $V(x) - V(y)$ is the difference in quasi-potential between each $x$ and $y$, and
\begin{equation}
    f^\mathrm{nl}(x,y) = -f^\mathrm{nl}(y,x) = \beta^{-1}\log \left(\frac{\rho_\mathrm{ss}(x)k(x,y)}{\rho_\mathrm{ss}(y)k(y,x)} \right).
\end{equation}
is the external driving to the nonlocal transitions between each $x$ and $y$. 

Recall that the work $\delta W = \text{force} \times \delta x$, where $x$ is the position, $\delta$ is the variation or derivative in general senses. Then based on the probability currents $j^\mathrm{loc}(t,x)$, and $j^\mathrm{nl}(t,x,y)$, the change of total work $dW(t)/dt$ done by the thermodynamic force $f^\mathrm{loc}(x)$ and external driving $f^{\text{nl}}(x,y)$ is given by:
\begin{align*}
    \frac{dW(t)}{dt}
    =& \int_{\mathbb{R}^{n}} (j^\mathrm{loc}(t,x))^\mathrm{T} f^\mathrm{loc}(x) dx + \frac{1}{2} \int_{\mathbb{R}^{n}} \int_{\mathbb{R}^{n} \backslash \{x\}} j^\mathrm{nl}(t,x,y) f^\mathrm{nl}(x,y) dydx\\
    =& \int_{\mathbb{R}^{n}} \beta^{-1}(j^\mathrm{loc}(t,x))^\mathrm{T}( \beta A^{-1}(x)b(x) - \nabla\log\rho_\mathrm{ss}(x) )dx\\
    &+ \frac{1}{2}\int_{\mathbb{R}^{n}} \int_{\mathbb{R}^{n} \backslash \{x\}} j^\mathrm{nl}(t,x,y) \beta^{-1}\log \left(\frac{\rho_\mathrm{ss}(x)k(x,y)}{\rho_\mathrm{ss}(y)k(y,x)} \right)dx\\
    =& Q_\mathrm{hk}(t).
\end{align*}
\begin{remark}
   The identity $dW(t)/dt = Q_\mathrm{hk}(t)$ establishes that the housekeeping heat corresponds to the power input by the thermodynamic forces in the general non-equilibrium setting. When thermodynamic forces are absent $f^\mathrm{loc} = 0,  f^\mathrm{nl} = 0$, the system satisfies detailed balance with respect to the Gibbs measure. In this case, $dW/dt = 0$ implies $Q_\mathrm{hk}(t) = 0$, and the free energy dissipation reduces to $dF/dt = -\beta^{-1} e_\mathrm{p}(t)$, recovering the classical result for general Langevin dynamics \cite{ZL25}.
\end{remark}

In the next theorem, we show that the housekeeping heat $Q_\mathrm{hk}(t)$, entropy production rate $e_\mathrm{p}(t)$ are always nonnegative, and the free energy dissipation $dF(\rho(t))/dt$ is always nonpositive.

\begin{thm}
For every $t>0$, the housekeeping heat $Q_\mathrm{hk}(t)$ and entropy production rate $e_\mathrm{p}(t)$ satisfy
\begin{equation*}
    Q_\mathrm{hk}(t) \geq 0, \quad e_\mathrm{p}(t) \geq 0.
\end{equation*}
Moreover, the generalized free energy satisfies
\begin{equation*}
    \frac{dF(\rho(t))}{dt} \leq 0.
\end{equation*}
\end{thm}

\begin{proof}
For the housekeeping heat $Q_\mathrm{hk}(t)$, since $\log r \leq r-1$, we have
\begin{align*}
    Q_\mathrm{hk}(t) \geq & \frac{1}{\beta}\int_{\mathbb{R}^n} [\beta A^{-1}(x)b(x) - \nabla \log\rho_\mathrm{ss}(x)]^\mathrm{T}\beta^{-1}A(x)[\beta A^{-1}(x)b(x) - \nabla \log\rho_\mathrm{ss}(x) ] \rho(t,x)dx \\
    & + \frac{1}{\beta}\int_{\mathbb{R}^n} [\nabla \log\rho_\mathrm{ss}(x) -  \nabla \log\rho(t,x)]^\mathrm{T}\beta^{-1}A(x)[\beta A^{-1}(x)b(x) - \nabla \log\rho_\mathrm{ss}(x) ] \rho(t,x)dx  \\
    & + \frac{1}{\beta}\int_{\mathbb{R}^n}\int_{\mathbb{R}^n \setminus \{x\}} \rho(t,x)k(x,y) \left( \frac{\rho_\mathrm{ss}(y) k(y,x)}{\rho_\mathrm{ss}(x)k(x,y)} -1\right) dydx.
\end{align*}
Using integration by parts, we have
\begin{align*}
    & \int_{\mathbb{R}^n} [\nabla \log\rho_\mathrm{ss}(x) -  \nabla \log\rho(t,x)]^\mathrm{T}\beta^{-1}A(x)[\beta A^{-1}(x)b(x) - \nabla \log\rho_\mathrm{ss}(x) ] \rho(t,x)dx\\
    = & -\int_{\mathbb{R}^n} \left[\nabla \log \left( \frac{\rho(t,x)}{\rho_\mathrm{ss}(x)} \right) \right]^\mathrm{T} [b(x)\rho_\mathrm{ss}(x) - \beta^{-1}A(x)\nabla\rho_\mathrm{ss}(x)]\frac{\rho(t,x)}{\rho_\mathrm{ss}(x)}dx \\
    = & \int_{\mathbb{R}^n} \nabla[b(x)\rho_\mathrm{ss}(x) - \beta^{-1}A(x)\nabla \rho_\mathrm{ss}(x)] \frac{\rho(t,x)}{\rho_\mathrm{ss}(x)}dx.
\end{align*}
For the nonlocal term, we have
\begin{align*}
    & \int_{\mathbb{R}^n}\int_{\mathbb{R}^n \setminus \{x\}} \rho(t,x)k(x,y) \left( \frac{\rho_\mathrm{ss}(y) k(y,x)}{\rho_\mathrm{ss}(x)k(x,y)} -1\right) dydx \\
    = & \int_{\mathbb{R}^n } \frac{\rho(t,x)}{\rho_\mathrm{ss}(x)} \int_{\mathbb{R}^n \setminus \{x\}}\rho_\mathrm{ss}(y) k(y,x) dydx - \int_{\mathbb{R}^n}\int_{\mathbb{R}^n \setminus \{x\}} \rho(t,x)k(x,y)dydx \\
    = & \int_{\mathbb{R}^n} \frac{\rho(t,x)}{\rho_\mathrm{ss}(x)} \left(\int_{\mathbb{R}^n \setminus \{x\}}\rho_\mathrm{ss}(y) k(y,x) - \rho_\mathrm{ss}(x)k(x,y)dy \right)dx.
\end{align*}
Recall that the stationary density satisfies the stationary Fokker-Planck equation, i.e. 
\begin{equation}\label{SnFKe}
    \nabla\cdot[b(x)\rho_\mathrm{ss}(x) - \beta^{-1}A(x)\nabla \rho_\mathrm{ss}(x)]+\int_{\mathbb{R}^n \setminus \{x\}}\rho_\mathrm{ss}(y) k(y,x) - \rho_\mathrm{ss}(x)k(x,y)dy  = 0.
\end{equation}
Combining above local and nonlocal terms,  if follows that
\begin{align*}
    Q_\mathrm{hk}(t) \geq & \frac{1}{\beta}\int_{\mathbb{R}^n} [A^{-1}(x)b(x) - \beta^{-1} \nabla \log\rho_\mathrm{ss}(x)]^\mathrm{T}\beta^{-1}A(x)[\beta A^{-1}(x)b(x) - \nabla \log\rho_\mathrm{ss}(x) ] \rho(t,x)dx \geq 0.
\end{align*}

For the entropy production rate $e_\mathrm{p}(t)$, note that $j^\mathrm{nl}(t,x,y)= \rho(t,x)k(x,y) - \rho(t,y)k(y,x)$ and $\log \left( \frac{\rho(t,x)k(x,y)}{\rho(t,y)k(y,x)} \right)$ are both positive or negative. Thus
\begin{equation*}
    e_\mathrm{p}(t)=\beta\int_{\mathbb{R}^n} \frac{(j^\mathrm{loc}(t,x))^\mathrm{T}A^{-1}(x)j^\mathrm{loc}(t,x) }{\rho(t,x)}dx+\int_{\mathbb{R}^n}\int_{\mathbb{R}^n \setminus \{x\}} j^\mathrm{nl}(t,x,y)\log \left( \frac{\rho(t,x)k(x,y)}{\rho(t,y)k(y,x)} \right) dydx \geq 0.
\end{equation*}
Moreover, this equality holds if and only if $j^\mathrm{nl}(t,x,y)= \rho(t,x)k(x,y) - \rho(t,y)k(y,x) = 0$ for every $x,y \in \mathbb{R}^n$ and every $t \geq 0$, and $j^\mathrm{loc}(t,x) = 0$ for every $x\in \mathbb{R}^n$ and every $t \geq 0$.

For the free energy dissipation, using the stationary Fokker-Planck equation \eqref{SnFKe} again, we have
\begin{align*}
    \frac{dF(t)}{dt}
    \leq  & - \frac{1}{\beta}\int_{\mathbb{R}^n} [\nabla \log\rho_\mathrm{ss}(x) - \nabla \log\rho(t,x) ]^\mathrm{T}\beta^{-1}A(x)(\nabla \log\rho_\mathrm{ss}(x) - \nabla \log\rho(t,x) ) \rho(t,x)dx \\
    &  - \frac{1}{2\beta} \int_{\mathbb{R}^n}\int_{\mathbb{R}^n \setminus \{x\}} [\rho(t,x)k(x,y) - \rho(t,y)k(y,x)]\log \left( \frac{\rho(t,x)k(x,y)}{\rho(t,y)k(y,x)} \right) dydx  \\
    & + \frac{1}{2\beta}\int_{\mathbb{R}^n} \int_{\mathbb{R}^n \setminus \{x\}} j^\mathrm{nl}(t,x,y)\log \left( \frac{k(x,y)\rho_\mathrm{ss}(y)}{k(y,x)\rho_\mathrm{ss}(x)}\right)dy dx\\
    \leq  & - \frac{1}{\beta}\int_{\mathbb{R}^n} [\nabla \log\rho(t,x) - \nabla \log\rho_\mathrm{ss}(x) ]^\mathrm{T}\beta A(x)(\nabla \log\rho(t,x) - \nabla \log\rho_\mathrm{ss}(x) ) \rho(t,x)dx  \\
    & - \frac{1}{2\beta} \int_{\mathbb{R}^n}\int_{\mathbb{R}^n \setminus \{x\}} j^\mathrm{nl}(t,x,y) \log \left( \frac{\rho(t,x)k(x,y)}{\rho(t,y)k(y,x)} \right) dydx  \\
    & + \frac{1}{4\beta}\int_{\mathbb{R}^n} \int_{\mathbb{R}^n \setminus \{x\}} j^\mathrm{nl}(t,x,y)\left[ \log \left( \frac{k(x,y)\rho_\mathrm{ss}(y)}{k(y,x)\rho_\mathrm{ss}(x)}\right) + \log \left( \frac{k(y,x)\rho_\mathrm{ss}(x)}{k(x,y)\rho_\mathrm{ss}(y)}\right) \right]dy dx. \\
    \leq & 0.
\end{align*}
Here, we use the fact that 
\begin{equation*}
     \log \left( \frac{k(x,y)\rho_\mathrm{ss}(y)}{k(y,x)\rho_\mathrm{ss}(x)}\right) + \log \left( \frac{k(y,x)\rho_\mathrm{ss}(x)}{k(x,y)\rho_\mathrm{ss}(y)}\right) = 0.
\end{equation*}
This implies the dissipative property of the free energy.
\end{proof}

Our result reveals a fundamental thermodynamic structure underlying general jump-diffusion processes, 
extending the classical understanding of free energy dissipation in Langevin systems to a broader class of stochastic dynamics with jumps. 
In next remark, we discuss the connection between our decomposition of free energy dissipation and the adiabatic/non-adiabatic decomposition in stochastic thermodynamics.

\begin{remark}
The decomposition of the free energy dissipation rate
\begin{equation}\label{eq:decomp_entropy}
  -\frac{dF(\rho(t))}{dt} = \beta^{-1} e_\mathrm{p}(t) - Q_\mathrm{hk}(t)
\end{equation}
can be structurally aligned with the canonical decomposition of entropy production in stochastic thermodynamics, commonly referred to as the adiabatic and non-adiabatic split \cite{esposito2010,EV10,VE10}.
Multiplying \eqref{eq:decomp_entropy} by $\beta$, we obtain the Clausius-type inequality
\begin{equation}
    e_\mathrm{p}(t) = \beta Q_\mathrm{hk}(t) + \left( -\beta \frac{dF(\rho(t))}{dt} \right) \geq 0.
\end{equation}
Here, the total entropy production rate $e_\mathrm{p}(t)$ corresponds to the variation of the total entropy $\dot{S}_{tot}$. The term associated with the free energy decay, $-\beta \dot{F}(t)$, represents the \emph{non-adiabatic entropy production rate} (often denoted as $\dot{S}_{na}$), which quantifies the dissipation due to the relaxation of the system towards its steady state $\rho_\mathrm{ss}$.
Consequently, the term $\beta Q_\mathrm{hk}(t)$ corresponds to the \emph{adiabatic entropy production rate} (denoted as $\dot{S}_{ad}$), which characterizes the dissipation required to maintain the non-equilibrium steady state against the non-conservative forces.
Specifically, $Q_\mathrm{hk}(t)$ provides the explicit expression for the housekeeping heat involved in the nonlocal transport of probability mass.
\end{remark}



\section{A decomposition of diffusion process}\label{Sec4}
The analysis in the previous section has shown that, compared to a Langevin system, the free energy dissipation rate of system $(X_t)_{t\ge0}$ under external forcing includes not only the entropy production rate but also an additional housekeeping heat term. This leads us to consider whether, there exists a subsystem of $(X_t)_{t\ge0}$, whose free energy dissipation is entirely characterized by a certain form of entropy production, as in the Langevin case.

It is known that a Langevin system with a symmetric generator $\mathcal{G}$ in $D(\mathcal{G}) \cap L^2(\mathbb{R}^n, \mu)$, satisfies detailed balance and time-reversal symmetry at its stationary measure $\mu$
\begin{equation*}
    \int_{\mathbb{R}^n} f(x) \mathcal{G}g(x) \mu(dx) = \int_{\mathbb{R}^n} g(x) \mathcal{G} f(x) \mu(dx).
\end{equation*}
Building on this structure, we further define the weighted Hilbert space $L^2(\mathbb{R}^n, \rho_\mathrm{ss}(x) dx)$ equipped with the inner product
\begin{equation*}
    \langle \varphi, \psi \rangle = \int_{\mathbb{R}^n} \varphi(x) \psi(x)  \rho_\mathrm{ss}(x)  dx, \quad \forall \varphi, \psi \in C_\mathrm{c}^\infty(\mathbb{R}^n),
\end{equation*}
which is central to analyzing the generator's adjoint structure under the stationary measure. This inner product induces the $\rho_\mathrm{ss}$-weighted adjoint $\mathcal{L}^\dagger$ of the generator $\mathcal{L}$, defined by the relation:
\begin{equation*}
    \langle \mathcal{L} \varphi, \psi \rangle = \langle \varphi, \mathcal{L}^\dagger \psi \rangle.
\end{equation*}
Therefore, our objective is to identify a component within $\mathcal{L}$ that resembles a Langevin system—specifically, its symmetric part. This leads us to the following theorem:
\begin{thm}
    The generator of $X_t$ can decompose as a symmetric operator and an anti-symmetric operator, that is, $\mathcal{L}=\mathcal{L}_\mathrm{s} + \mathcal{L}_\mathrm{a}$, where
    \begin{align}\label{op:Ls}
        \mathcal{L}_\mathrm{s}\varphi
        =& \beta^{-1}[A(x)\nabla\log\rho_\mathrm{ss}(x)]^\mathrm{T}\nabla\varphi(x) + \beta^{-1} \nabla\cdot (A(x)\nabla \varphi(x)) \nonumber \\
        &+ \frac{1}{2}\int_{\mathbb{R}^{n}\setminus\{x\}} \Big( \varphi(y)-\varphi(x) \Big) \left( \frac{k(x,y)}{\rho_\mathrm{ss}(y)} + \frac{k(y,x)}{\rho_\mathrm{ss}(x)} \right)\rho_\mathrm{ss}(y) dy,
    \end{align}
and    
\begin{align}\label{op:La}
     \mathcal{L}_\mathrm{a}\varphi=
    & \Big[ b(x)  - \beta^{-1}A(x)\nabla\log\rho_\mathrm{ss}(x) \Big]^\mathrm{T}\nabla \varphi(x) \nonumber \\
    &+ \frac{1}{2}\int_{\mathbb{R}^{n}\setminus\{x\}} \Big( \varphi(y)-\varphi(x) \Big) \left( \frac{k(x,y)}{\rho_\mathrm{ss}(y)} - \frac{k(y,x)}{\rho_\mathrm{ss}(x)} \right)\rho_\mathrm{ss}(y) dy.
\end{align}
\end{thm}
\begin{proof}
Recall that $\mathcal{L}\varphi(x) = \mathcal{L}_\mathrm{loc}\varphi(x) + \mathcal{L}_\mathrm{nl}\varphi(x)$ and $\mathcal{L}^{\ast}\varphi(x) = \mathcal{L}^{\ast}_\mathrm{loc}\varphi(x) + \mathcal{L}^{\ast}_\mathrm{nl}\varphi(x) = -\nabla j^\mathrm{loc}(t,x)-\int j^\mathrm{nl}(t,x,y)$, where $\mathcal{L}^{\ast}$ is the adjoint operator of $\mathcal{L}$ under usual inner product, now we compute the local part of $\mathcal{L}^\dagger$
\begin{align*}
    \langle \mathcal{L}_\mathrm{loc}\varphi, \psi \rangle 
    &= \int_{\mathbb{R}^{n}}  \Big\{ b^\mathrm{T}(x) \nabla \varphi(x) + \beta^{-1} \nabla\cdot (A(x)\nabla \varphi(x)) \Big\}\rho_\mathrm{ss}(x)\psi(x) dx\\
    &= \int_{\mathbb{R}^{n}} \Big\{ \nabla\cdot \Big[-b(x)\rho_\mathrm{ss}(x)\psi(x) + \beta^{-1} A(x)\nabla (\rho_\mathrm{ss}(x)\psi(x)) \Big]  \Big\} \varphi(x) dx \\
    &= \langle \varphi, \rho_\mathrm{ss}^{-1}\mathcal{L}_\mathrm{loc}^{*} ( \rho_\mathrm{ss}\psi ) \rangle := \langle \varphi, \mathcal{L}_\mathrm{loc}^\dagger \psi \rangle.
\end{align*}
And the non-local part
\begin{align*}
    \langle \mathcal{L}_\mathrm{nl}\varphi, \psi \rangle
    &= \int_{\mathbb{R}^n}\psi(x)\rho_\mathrm{ss}(x)\int_{\mathbb{R}^n \setminus \{ x \} } [\varphi(y) - \varphi(x)]k(x,y) dydx\\
    &= \int_{\mathbb{R}^n}\int_{\mathbb{R}^n \setminus \{ x \} }\psi(x)\rho_\mathrm{ss}(x)\varphi(y)k(x,y)dydx - \int_{\mathbb{R}^n}\int_{\mathbb{R}^n \setminus \{ x \} }\psi(x)\rho_\mathrm{ss}(x)\varphi(x)k(x,y)dydx\\
    &= \int_{\mathbb{R}^n} \varphi(y) \int_{\mathbb{R}^n \setminus \{ y \} }\psi(x)\rho_\mathrm{ss}(x)k(x,y)dxdy - \int_{\mathbb{R}^n} \varphi(x)\rho_\mathrm{ss}(x) \int_{\mathbb{R}^n \setminus \{ x \} }\psi(x)k(x,y)dydx\\
    &= \int_{\mathbb{R}^n} \varphi(x)\rho_\mathrm{ss}(x)\rho_\mathrm{ss}^{-1}(x) \int_{\mathbb{R}^n \setminus \{ x \} }\psi(y)\rho_\mathrm{ss}(y)k(y,x)dydx - \int_{\mathbb{R}^n} \varphi(x)\rho_\mathrm{ss}(x) \int_{\mathbb{R}^n \setminus \{ x \} }\psi(x)k(x,y)dydx\\
    &= \left\langle \varphi, \rho_\mathrm{ss}^{-1}(x)\int_{\mathbb{R}^n \setminus \{ x \} } \Big[\rho_\mathrm{ss}(y)\psi(y)k(y,x)-\rho_\mathrm{ss}(x)\psi(x)k(x,y) \Big] dy \right\rangle \\
    &= \langle \varphi , \rho_\mathrm{ss}^{-1}\mathcal{L}_\mathrm{nl}^{*} ( \rho_\mathrm{ss}\psi ) \rangle := \langle \varphi, \mathcal{L}_\mathrm{nl}^\dagger \psi \rangle.
\end{align*}
So we have
\begin{align*}
    \mathcal{L}^\dagger \varphi(x)
    =& \mathcal{L}_\mathrm{loc}^\dagger \varphi(x) + \mathcal{L}_\mathrm{nl}^\dagger \varphi(x) = \rho_\mathrm{ss}^{-1}\mathcal{L}_\mathrm{loc}^{*} ( \rho_\mathrm{ss}\varphi ) + \rho_\mathrm{ss}^{-1}\mathcal{L}_\mathrm{nl}^{*} ( \rho_\mathrm{ss}\varphi ) \\
    =& \rho_\mathrm{ss}^{-1}(x)\nabla\cdot \Big[-b(x)\rho_\mathrm{ss}(x)\varphi(x) + \beta^{-1} A(x)\nabla (\rho_\mathrm{ss}(x)\varphi(x)) \Big]\\
    &+ \rho_\mathrm{ss}^{-1}(x)\int_{\mathbb{R}^n \setminus \{ x \} } \Big[\rho_\mathrm{ss}(y)\varphi(y)k(y,x)-\rho_\mathrm{ss}(x)\varphi(x)k(x,y) \Big] dy\\
    =& \rho_\mathrm{ss}^{-1}\mathcal{L}^{*} ( \rho_\mathrm{ss}\varphi ).
\end{align*}
The symmetric operator and the anti-symmetric operator are defined as
\begin{gather*}
    \mathcal{L}_\mathrm{s}\varphi = \frac{1}{2}\left[\mathcal{L}\varphi + \mathcal{L}^\dagger\varphi\right], \quad    \mathcal{L}_\mathrm{a}\varphi = \frac{1}{2}\left[\mathcal{L}\varphi - \mathcal{L}^\dagger\varphi\right].
\end{gather*}
Then direct calculation shows that the symmetric operator $\mathcal{L}_\mathrm{s}$ is given by (\ref{op:Ls}), and the anti-symmetric operator $\mathcal{L}_\mathrm{a}$ is given by (\ref{op:La}).
\end{proof}

Using the decomposition of the generator, we can express the free energy dissipation rate as
\begin{align*}
    \frac{dF(\rho(t))}{dt}
    =& \beta^{-1}\int_{\mathbb{R}^n} \mathcal{L}^\ast \rho(t,x) \left(\log \left( \frac{\rho(t,x)}{\rho_\mathrm{ss}(x)} \right) +1 \right) dx \\
    =& \beta^{-1}\int_{\mathbb{R}^n} \big[ (\mathcal{L}_\mathrm{a})^*\rho(t,x) + (\mathcal{L}_\mathrm{s})^*\rho(t,x) \big]\left(\log \left( \frac{\rho(t,x)}{\rho_\mathrm{ss}(x)} \right) +1 \right) dx\\
    :=& F'_\mathrm{a}(\rho(t)) + F'_\mathrm{s}(\rho(t)).
\end{align*}
In what follows, we will further compute the respective contributions of these two components to the overall dissipation and examine their relationship with both the entropy production rate and the housekeeping heat.

\subsection{Canonical conservative dynamics }

For the anti-symmetric operator $\mathcal{L}_{a}$, we introduce 
\begin{equation*}
    v_\mathrm{a}(x) =  b(x)  - \beta^{-1}A(x)\nabla\log\rho_\mathrm{ss}(x), \quad k_\mathrm{a}(x,y) = \frac{1}{2}\left( \frac{k(x,y)}{\rho_\mathrm{ss}(y)} - \frac{k(y,x)}{\rho_\mathrm{ss}(x)} \right)\rho_\mathrm{ss}(y), \ \forall x,y \in \mathbb{R}^n. 
\end{equation*}
Then the anti-symmetric operator could be rewritten as
\begin{align*}
    \mathcal{L}_\mathrm{a}\varphi
    :=& v_\mathrm{a}^\mathrm{T}(x)\nabla \varphi(x) + \int_{\mathbb{R}^{n}\setminus\{x\}} \Big( \varphi(y)-\varphi(x) \Big)k_\mathrm{a}(x,y)dy,
\end{align*}
and the Fokker-Planck (adjoint) operator associated with $\mathcal{L}_\mathrm{a}$ is
\begin{align*}
    (\mathcal{L}_\mathrm{a})^*\varphi
    =& -\nabla\cdot (v_\mathrm{a}(x)\varphi(x)) + \int_{\mathbb{R}^n \setminus \{ x \} } \Big[ \varphi(y)k_\mathrm{a}(y,x) - \varphi(x)k_\mathrm{a}(x,y) \Big] dy.
\end{align*}
The dynamics generated by $\mathcal{L}_\mathrm{a}$ is a canonical conservative dynamics
\begin{equation*}
    \partial_t \rho_\mathrm{a}(t,x) = (\mathcal{L}_{a})^* \rho_\mathrm{a}(t,x).
\end{equation*}
The associated local and nonlocal probability current are given by
\begin{align*}
    j^\mathrm{loc}_\mathrm{a}(t,x) =  -v_\mathrm{a}(x)\rho(t,x),
\end{align*}
and 
\begin{align*}
    j^\mathrm{nl}_\mathrm{a}(t,x,y)
    =& \rho(t,x)k_\mathrm{a}(x,y) - \rho(t,y)k_\mathrm{a}(y,x)\\
    =& \frac{1}{2}\left( \frac{k(x,y)}{\rho_\mathrm{ss}(y)} - \frac{k(y,x)}{\rho_\mathrm{ss}(x)} \right) \Big( \rho_\mathrm{ss}(y)\rho(t,x) + \rho_\mathrm{ss}(x)\rho(t,y) \Big).
\end{align*}
The canonical conservative dynamics (current) preserves the generalized free energy, we will show this property in the following lemma and proposition.

\begin{lemma}
The Fokker-Planck operator satisfies $(\mathcal{L}_\mathrm{a})^*\rho_\mathrm{ss}(x)=0$.\\
Moreover, the drift $v_\mathrm{a}$ satisfies 
\begin{equation}\label{anti-va}
    \nabla\cdot (v_\mathrm{a}(x)\rho_\mathrm{ss}(x)) = (\nabla v_\mathrm{a}(x) + v_\mathrm{a}(x)\nabla\log \rho_\mathrm{ss}(x))\rho_\mathrm{ss}(x) = 0
\end{equation}
and $k_{a}(x,y)$ vanishes upon integration over $y$ i.e.
        \begin{equation}\label{ka-0}
            \int_{\mathbb{R}^n \setminus \{x\}} k_{a}(x,y)dy = 0,\quad \forall x \in \mathbb{R}^n.
        \end{equation}
\end{lemma}

\begin{proof}
Note that $\nabla j^\mathrm{loc}_\mathrm{ss}(x) = \nabla [(b(x) - \beta^{-1}A(x) \nabla \log \rho_\mathrm{ss}(x))\rho_\mathrm{ss}(x)] = 0$, and
\begin{equation*}
    \beta^{-1}A(x)\nabla\log\rho_\mathrm{ss}(x)\cdot\rho_\mathrm{ss}(x) - \beta^{-1}A(x)\nabla\rho_\mathrm{ss}(x) = 0.
\end{equation*}
It follows that
\begin{equation*}
    \nabla[v_\mathrm{a}(x)\rho_\mathrm{ss}(x)] = \nabla[b(x)\rho_\mathrm{ss}(x)] - \beta^{-1} \nabla[A(x)\nabla\rho_\mathrm{ss}(x)] = 0.
\end{equation*}
For the nonlocal part, by the definition of $k_\mathrm{a}(x,y)$, we have
\begin{equation*}
    \rho_\mathrm{ss}(x)k_\mathrm{a}(x,y) - \rho_\mathrm{ss}(y)k_\mathrm{a}(y,x) = \rho_\mathrm{ss}(x)k(x,y) - \rho_\mathrm{ss}(y)k(y,x), \quad \forall x\neq y
\end{equation*}
Then we have $(\mathcal{L}_{a})^*\rho_\mathrm{ss}(x) = 0$. 

Since $\rho_\mathrm{ss}$ satisfies
\begin{equation*}
    \int_{\mathbb{R}^n \setminus \{x\}} k(x,y)\rho_\mathrm{ss}(x) - k(y,x)\rho_\mathrm{ss}(y) dy = 0,
\end{equation*}
combining with the definition of $k_\mathrm{a}(x,y)$ yields
\begin{align}\label{anti-nonlocal}
  \int_{\mathbb{R}^n \setminus \{x\}}\rho_\mathrm{ss}(x) k_\mathrm{a}(x,y) dy = 0.
\end{align}
Then we have
\begin{equation}
    \int_{\mathbb{R}^n \setminus \{x\}}\rho_\mathrm{ss}(y) k_\mathrm{a}(y,x) dy = \int_{\mathbb{R}^n \setminus \{x\}} \rho_\mathrm{ss}(x)k_\mathrm{a}(x,y)dy = \rho_\mathrm{ss}(x)\int_{\mathbb{R}^n \setminus \{x\}} k_\mathrm{a}(x,y)dy = 0.
\end{equation}
Since the density function is positive, we obtain \eqref{ka-0}.
\end{proof}

\begin{prop}\label{prop:La}
The free energy dissipation associated with the anti-symmetric part is zero, i.e.
\begin{equation}\label{eq:Fadiss}
    F'_\mathrm{a}(\rho(t)) = \beta^{-1}\int_{\mathbb{R}^n} \mathcal{L}_\mathrm{a}^\ast \rho(t,x) \left(\log \left( \frac{\rho(t,x)}{\rho_\mathrm{ss}(x)} \right) +1 \right) dx = 0.
\end{equation}
\end{prop}
\begin{proof}
The free energy dissipation is 
\begin{align*}
   F'_\mathrm{a}(\rho(t,x))  = & \frac{1}{\beta}\int_{\mathbb{R}^n}[\rho(t,x) v_\mathrm{a}(x)]^\mathrm{T} [\nabla \log\rho_\mathrm{ss}(x) - \nabla \log\rho(t,x)] dx \\
    & +\frac{1}{\beta}\int_{\mathbb{R}^n}\int_{\mathbb{R}^n \setminus \{x\}} [\rho(t,y)k_\mathrm{a}(y,x) - \rho(t,x)k_\mathrm{a}(x,y)]\log \left( \frac{\rho(t,x) }{\rho_\mathrm{ss}(x)} \right) dydx \\
    = & I_1 + I_2. 
\end{align*}
For $I_1$, by \eqref{anti-va}, we have
\begin{align*}
    I_1 
    = & \frac{1}{\beta}\int_{\mathbb{R}^n}[\rho(t,x) v_\mathrm{a}(x)]^\mathrm{T} (\nabla \log\rho_\mathrm{ss}(x) - \nabla \log\rho(t,x)) dx\\
   =& \frac{1}{\beta}\int_{\mathbb{R}^n} \Big(\nabla\cdot v_\mathrm{a}(x) + v_\mathrm{a}(x) \nabla\log \rho_\mathrm{ss}(x)\Big) \rho(t,x) dx\\
    = &  0.
\end{align*}
For $I_2$, we have
\begin{align*}
     I_2= & \frac{1}{2\beta}\int_{\mathbb{R}^n}\int_{\mathbb{R}^n \setminus \{x\}} \rho(t,x)k_\mathrm{a}(x,y) \left[  \log \left( \frac{\rho(t,y) }{\rho_\mathrm{ss}(y)} \right) - \log \left( \frac{\rho(t,x) }{\rho_\mathrm{ss}(x)} \right) \right] dydx \\
     =& \frac{1}{2\beta}\int_{\mathbb{R}^n}\int_{\mathbb{R}^n \setminus \{x\}} \rho_\mathrm{ss}(x)k_\mathrm{a}(x,y) \frac{\rho(t,x)}{\rho_\mathrm{ss}(x)} \left[  \log \left( \frac{\rho(t,y) }{\rho_\mathrm{ss}(y)} \right) - \log \left( \frac{\rho(t,x) }{\rho_\mathrm{ss}(x)} \right) \right] dydx \nonumber\\
     =& \frac{1}{4\beta}\int_{\mathbb{R}^n}\int_{\mathbb{R}^n \setminus \{x\}} \rho_\mathrm{ss}(x)k_\mathrm{a}(x,y) \Bigg(\frac{\rho(t,x)}{\rho_\mathrm{ss}(x)}   \left[  \log \left( \frac{\rho(t,y) }{\rho_\mathrm{ss}(y)} \right) - \log \left( \frac{\rho(t,x) }{\rho_\mathrm{ss}(x)} \right) \right] \\
     &- \frac{\rho(t,y)}{\rho_\mathrm{ss}(y)}   \left[ \log \left( \frac{\rho(t,x) }{\rho_\mathrm{ss}(x)} \right) - \log \left( \frac{\rho(t,y) }{\rho_\mathrm{ss}(y)} \right) \right] \Bigg) dydx \nonumber\\
     =& \frac{1}{4\beta}\int_{\mathbb{R}^n}\int_{\mathbb{R}^n \setminus \{x\}} \rho_\mathrm{ss}(x)k_\mathrm{a}(x,y) \left[ \frac{\rho(t,x)}{\rho_\mathrm{ss}(x)} + \frac{\rho(t,y)}{\rho_\mathrm{ss}(y)} \right]  \left[  \log \left( \frac{\rho(t,y) }{\rho_\mathrm{ss}(y)} \right) - \log \left( \frac{\rho(t,x) }{\rho_\mathrm{ss}(x)} \right) \right] dydx.
\end{align*}
By \eqref{ka-0}, we also have
\begin{align}
    I_2 =& \frac{1}{\beta}\int_{\mathbb{R}^n}\int_{\mathbb{R}^n \setminus \{x\}} \rho(t,y)k_\mathrm{a}(y,x) \log \left( \frac{\rho(t,x) }{\rho_\mathrm{ss}(x)} \right) dydx \nonumber\\
    = &  \frac{1}{2\beta}\int_{\mathbb{R}^n}\int_{\mathbb{R}^n \setminus \{x\}} \rho_\mathrm{ss}(y)k_\mathrm{a}(y,x) \frac{\rho(t,x)}{\rho_\mathrm{ss}(x)} \left[ \log \left( \frac{\rho(t,x)}{\rho_\mathrm{ss}(x) } \right) - \log \left( \frac{\rho(t,y)}{\rho_\mathrm{ss}(y)} \right) \right] dydx \nonumber \\
    = &   \frac{1}{4\beta}\int_{\mathbb{R}^n}\int_{\mathbb{R}^n \setminus \{x\}} \rho_\mathrm{ss}(x)k_\mathrm{a}(x,y) \left[ \frac{\rho(t,y)}{\rho_\mathrm{ss}(y)} + \frac{\rho(t,x)}{\rho_\mathrm{ss}(x)}  \right]  \left[ \log \left( \frac{\rho(t,x)}{\rho_\mathrm{ss}(x) } \right) - \log \left( \frac{\rho(t,y)}{\rho_\mathrm{ss}(y) } \right) \right] dydx.
\end{align}
Combining with above two equalities, we have $I_2 = -I_2$, i.e. $I_2 = 0$. Thus we finally obtain \eqref{eq:Fadiss}.
\end{proof}

The vanishing of the free energy dissipation, $F'_\mathrm{a}(\rho) = 0$, highlights the physical role of the anti-symmetric generator $\mathcal{L}_\mathrm{a}$. 
The operator $\mathcal{L}_\mathrm{a}$ drives the system in a ceaseless motion that preserves the stationary measure $\rho_\mathrm{ss}$ dynamically, 
so that the free energy remains constant due to the associated canonical conservative flow.
This property analogous to the inertial motion in Hamiltonian mechanics, thereby creating circulation without dissipation.

\begin{remark}
With the non-equilibrium steady state $\rho_\mathrm{ss}$, stationarity is maintained not by the detailed balance, but by the condition that the net flux along any closed cycle sums to zero. 
The vector field $v_\mathrm{a}(x)$ drives the probability mass transport along the level sets of the stationary density $\rho_\mathrm{ss}$, 
The jump kernel $k_\mathrm{a}(x,y)$ explicitly constructs these non-vanishing, divergence-free probability circulations (loops like $x \to y \to z \to x$).
Thus the canonical conservative flow creates a circulation flow that contributes to the breaking of detailed balance without altering the stationary density $\rho_\mathrm{ss}$ itself.
\end{remark}

\subsection{Symmetric reversible diffusion}

The symmetric operator $\mathcal{L}_\mathrm{s}$ can be served as an infinitesimal generator for a jump-diffusion process $X^\mathrm{sym}_t$, which is governed by the SDE
\begin{equation}
    dX^\mathrm{sym}_t = \Big[ \beta^{-1}A(X^\mathrm{sym}_t)\nabla\log\rho_\mathrm{ss}(X^\mathrm{sym}_t) + \beta^{-1} \nabla\cdot A(X^\mathrm{sym}_t) \Big] dt + \sqrt{2\beta^{-1}}a(X^\mathrm{sym}_t)dB_t + dJ_t.
\end{equation}
where the jump part $dJ_t$ is a Poisson random measure with the jump kernel
\begin{equation*}
    k_\mathrm{s}(x,y) = \frac{1}{2}\left( \frac{k(x,y)}{\rho_\mathrm{ss}(y)} + \frac{k(y,x)}{\rho_\mathrm{ss}(x)} \right)\rho_\mathrm{ss}(y).
\end{equation*}
Since the generator $\mathcal{L}_{s}$ is a self-adjoint operator with respect to the weighted $L^2$ space $L^2(\mathbb{R}^n, \rho_\mathrm{ss}(x)dx)$, 
by \cite[Theorem 5.1]{ZL25}, the jump-diffusion process $X^\mathrm{sym}_t$ is time-reversible with the stationary density $\rho_\mathrm{ss}$, and satisfies the detailed balance condition at the stationary density $\rho_\mathrm{ss}(x)$, i.e.
\begin{align*}
    \beta^{-1}A(x)\nabla\log\rho_\mathrm{ss}(x)\rho_\mathrm{ss}(x) - \beta^{-1}A(x)\nabla\rho_\mathrm{ss}(x) = 0,\\
    k_\mathrm{s}(x,y)\rho_\mathrm{ss}(x) - k_\mathrm{s}(y,x)\rho_\mathrm{ss}(y) = 0.
\end{align*}
Using the symmetric part $\mathcal{L}_\mathrm{s}$, we can define a symmetric reversible jump diffusion by using the following Dirichlet form associated with $\mathcal{L}_\mathrm{s}$.

\begin{definition}
We define the Dirichlet form of the symmetric generator $\mathcal{L}_\mathrm{s}$ as
\[
\mathcal{E}(f,g)
:= - \langle f, \mathcal{L}_\mathrm{s} g\rangle_\mathrm{ss}
= -\int_{\mathbb{R}^n} f(x)( \mathcal{L}_\mathrm{s} g)(x)\rho_\mathrm{ss}(x)\,dx, \quad \forall f,g\in L^2(\rho_\mathrm{ss}).
\]
Equivalently,
\[
\mathcal{E}(f,g)
= \int_{\mathbb{R}^n} \big( \nabla f(x) \big)^\mathrm{T} \nabla g(x) \rho_\mathrm{ss}(x)dx
+ \frac12 \int_{\mathbb{R}^n}  \int_{\mathbb{R}^n \setminus \{x\} }  \big(f(x)-f(y)\big)\big(g(x)-g(y)\big)k_\mathrm{s}(x,y)\rho_\mathrm{ss}(x) dy dx.
\]
\end{definition}

Now we define the Fisher information of the jump diffusions $X_t$ via the Dirichlet form of the symmetric generator $\mathcal{L}_\mathrm{s}$.

\begin{definition}
Let $f(t):= d\mu(t)/d\mu_\mathrm{ss} =\rho/\rho_\mathrm{ss}$ be the relative density. We define the Fisher information of the jump diffusions $X_t$ as
\[
\mathcal{I}[\rho(t)]
:= \mathcal{E}(f(t),\log f(t))
= -\langle f(t), \mathcal{L}_\mathrm{s} \log f(t)\rangle_\mathrm{ss}.
\]
\end{definition}
Note that the Fisher information $\mathcal{I}[\rho]$ induced by $\mathcal{L}_\mathrm{s}$ admits a local--nonlocal decomposition
\[
\mathcal{I}[\rho]
= \mathcal{I}_{\mathrm{loc}}[\rho] + \mathcal{I}_{\mathrm{nl}}[\rho],
\]
where $\mathcal{I}_{\mathrm{loc}}$ is the diffusion part and $\mathcal{I}_{\mathrm{nl}}$ is a nonlocal ``Carr\'e du champ'' built from the symmetric jump kernel. 

For the local part, since
$ \mathcal{L}^\mathrm{loc}_\mathrm{s}\log f
= \rho_\mathrm{ss}^{-1}\nabla\cdot(\rho_\mathrm{ss}A\nabla\log f)$, after integration by parts, we have
\begin{align*}
    \mathcal{I}_{\rm loc}[\rho(t)] = &  -\int_{\mathbb{R}^n} f \nabla\cdot(\rho_\mathrm{ss}A\nabla\log f)\,dx\\
    = & \int_{\mathbb{R}^n} (\nabla f)^\mathrm{T}\rho_\mathrm{ss}A\nabla\log f\,dx \\
 = & \int_{\mathbb{R}^n} f\rho_\mathrm{ss}(\nabla\log f)^\mathrm{T}A\nabla\log f dx \\ 
 = & \int_{\mathbb{R}^n} [(\nabla\log f(x))^\mathrm{T} A(x)\, \nabla\log f(x)]\rho(x)dx.
\end{align*}
For the nonlocal part,
\begin{align*}
\mathcal{I}_{\rm nl}[\rho(t)]
=  & \frac12\int_{\mathbb{R}^n}\int_{\mathbb{R}^n\setminus \{x\}}
k_\mathrm{s}(x,y)\rho_\mathrm{ss}(x)
\big(\log f(x)-\log f(y)\big)
\big(f(x)-f(y)\big) dydx \\
    = & \frac12 \int_{\mathbb{R}^n}\int_{\mathbb{R}^n\setminus \{x\}}
k_\mathrm{s}(x,y)\rho_\mathrm{ss}(x)\,
L\big(f(x),f(y)\big)\,
\big(\log f(x)-\log f(y)\big)^2\,
dydx.
\end{align*}
where
\[
L(u,v):=\frac{u-v}{\log u - \log v}
\]
is the logarithmic mean.  
This is the nonlocal analogue of a ``Carr\'e du champ'' form:
it represents a nonlocal squared gradient of $\log f$ weighted by the logarithmic mean.
The appearance of the logarithmic mean $L(u, v)$ is fundamental here. It serves as the discrete/nonlocal analog of the density weight $\rho(x)$ in the classical Fisher information. The identity $(u-v)(\log u - \log v) = L(u,v)(\log u - \log v)^2$ allows us to interpret the entropy production as a squared nonlocal gradient norm weighted by $L(\rho(x), \rho(y))$, ensuring that $\mathcal{I}_{\mathrm{nl}}$ converges to the standard Fisher information in the diffusion limit.

By Proposition \ref{prop:La}, the free energy dissipation associated with $\mathcal{L}_\mathrm{a}$ is $0$, and all free energy dissipation of the original jump diffusion is carried by $\mathcal{L}_\mathrm{s}$. 
Using the definition of the Fisher information, we have the following result.

\begin{prop}
The free energy dissipation of the jump-diffusion process $X_t$ is determined solely by the symmetric part $\mathcal{L}_\mathrm{s}$ of the generator, i.e.
\begin{equation*}
    \frac{d}{dt} F[\rho(t)] = \frac{d}{dt} F_\mathrm{s}[\rho(t)]
    = \beta^{-1}\langle f_t, L_\mathrm{s}\log f_t \rangle_\mathrm{ss}
    = -\beta^{-1}\mathcal{I}[\rho(t)]
    \le 0.
\end{equation*}
where $F_{s}[\rho(t)]$ is the free energy functional associated with $\mathcal{L}_\mathrm{s}$, and $\mathcal{I}$ is the Fisher information of $X_t$.
\end{prop}

The theoretical significance of the nonlocal Fisher information $\mathcal{I}[\rho]$ lies in its connection to the exponential convergence of the system. specifically, 
if the symmetric generator $\mathcal{L}_\mathrm{s}$ satisfies a Modified Log-Sobolev Inequality \cite{W00,C04}:
\begin{equation}\label{eq:MLSI}
    F(\rho(t)) = \beta^{-1}\mathcal{H}(\mu | \mu_\mathrm{ss}) \le \beta^{-1}C_{I} \mathcal{I}(\rho(t)),
\end{equation}
for some constant $C_{I} > 0$, 
then the free energy decays exponentially:
$$F(t) \le F(0) \mathrm{e}^{- \frac{t}{\beta C_{I}}}.$$
While establishing the precise conditions for \eqref{eq:MLSI} for general L\'{e}vy driven SDEs is a complex analytical problem depending on the tail behavior of the L\'{e}vy measure \cite{GI08}, 
the decomposition $\mathcal{L} = \mathcal{L}_\mathrm{s} + \mathcal{L}_\mathrm{a}$ clarifies that the relaxation rate is determined solely by the spectral properties of the symmetric part $\mathcal{L}_\mathrm{s}$.

\section{Application: intracellular particle transports}

In this section, we apply our non-Gaussian stochastic-thermodynamic framework to analyze some intracellular particle transports in biological systems operating far from equilibrium.
With a particular focus on active transport processes that involve external energy input, we provide a practical framework for their quantification and analytical characterization.

A confined Brownian particle is classically described by a Langevin equation with Gaussian noise (e.g.\ $dX_t^c=-\nabla V(X_t^c)\,dt+\sqrt{2}\,dB_t$). 
In crowded and compartmentalized intracellular environments, purely diffusive motion is often insufficient to account for observed intermittent long relocations. 
Cargos can transiently engage cytoskeletal tracks and molecular motors or be carried within vesicles, producing directed ``run'' phases superimposed on background diffusion \cite{vale1985,vale1985b,schroer1989}. 
Because such active processes are fueled by chemical energy, detailed balance is typically broken, and non-Gaussian jump  modeling becomes a natural coarse-grained choice \cite{schliwa1984,stenmark2009}. In particular, important evidence for employing a jump-process model comes from recent observations of intracellular vesicle transport \cite{chen2015}, which show that jump signatures in the data become more distinct as the average number $\langle n\rangle$ of molecular motors clustered around a vesicle increases.

We model the intracellular active transport dynamics by the jump diffusion introduced in Section~\ref{Sec2}:
\begin{equation*}
    dX_t =  b(X_{t})dt +  \sqrt{2\beta^{-1}} a(X_{t^-}) \circ dB_t  +  \int_{\mathbb{R}^n\setminus\{0\}} \sigma(X_{t^{-}},z) N(dt,dz),
\end{equation*}
Here $b$ encodes effective deterministic drift (e.g.\ confinement or mean reversion); the Brownian term represents thermal agitation and unresolved small-scale collisions; and the jump term represents intermittent burst-like relocations such as engagement/disengagement with motor-driven transport, abrupt release events, or switching between caged and running phases. The number or density of the underlying microscopic structures (e.g,\ the average motor count $\langle n\rangle$) constitutes the Poisson intensity rate $\lambda(\langle n\rangle,\dots)$. This rate directly governs the intensity of the non-Gaussian noise \cite{chen2015}.

Active transport generally violates detailed balance, while under structural conditions on $(b,a,\sigma)$ the model may admit an equilibrium steady state (see e.g. \cite{ZL25}). 
Hence the same framework can describe both equilibrium-like and non-equilibrium transport regimes. 
We view a stimulus as changing the effective coefficients and thus changing the target steady state.


For intracellular (jump) active transport model, we focus on: (i) relaxation toward the steady state, and (ii) the maintenance cost once a (possibly non-equilibrium) steady state is reached. 
These are quantified by the macroscopic observables from Section~\ref{Sec3}: free energy $F(t)$ (distance to the target $\mu_{\mathrm{ss}}$), entropy production rate $e_{\mathrm p}(t)$ (time-irreversibility), and housekeeping heat $Q_{\mathrm{hk}}(t)$ (steady non-equilibrium maintenance), linked by
\[
\frac{dF}{dt}=Q_{\mathrm{hk}}-\beta^{-1}e_{\mathrm p}.
\]
Moreover, the generator decomposition $\mathcal L=\mathcal L_{\mathrm s}+\mathcal L_{\mathrm a}$ from Section~\ref{Sec4} separates relaxation and circulation: $\mathcal L_{\mathrm s}$ determines the monotone decay of $F$ via a nonlocal Fisher-information/Dirichlet-form structure, while $\mathcal L_{\mathrm a}$ generates measure-preserving circulation (stationary currents) and thus captures non-equilibrium maintenance without contributing to free-energy decay.

\begin{remark}
    The SDE discussed in this section is limited to one dimension and employs a simplified form of drift. Although modeling can be performed in a high‑dimensional space, actual biological experiments often involve working with low‑dimensional data. For instance, in some experimental settings (e.g., axons or elongated cellular protrusions), transport effectively occurs along quasi-one-dimensional tracks, and microscopic imaging often yields projected two-dimensional trajectories.
    In practical applications, factors such as cell morphology, interactions among different chemical species, and the spatial distribution of experimental data all influence the selection of model parameters. Appropriate models should therefore be constructed based on the actual context, and methods such as Monte Carlo simulation can be employed for numerical investigation.
\end{remark}

\subsection{Intracellular active transport with non-equilibrium steady state}

We present a one-dimensional SDE as a toy model to describe the  particle transport within a single cell:
\begin{equation}\label{EX1SDE}
dX_t = -X_tdt + dB_t + dL^{\alpha}_t, 
\end{equation}
where $B_t$ is a one-dimensional standard Brownian motion and $L^{\alpha}_t$ is a one-dimensional $\alpha$-stable L\'{e}vy process with L\'{e}vy measure $\nu(dz) = C_{1,\alpha}|z|^{-(1+\alpha)}dz$, where $C_{1,\alpha}$ is a normalization constant. The notation of $\alpha$-stable L\'{e}vy process $dL^{\alpha}_t$ can also be seen as compound Poisson process $\int z N(dt,dz)$ with L\'{e}vy measure $\nu(dz) = C_{1,\alpha}|z|^{-(1+\alpha)}dz$. In the following, we demonstrate that \eqref{EX1SDE} provides an appropriate description of the intracellular active transport model. Let the initial distribution $Law(X_0)$ be a normal distribution with mean $3$ and variance $0.1$. Then the solution is given by a linear combination of the Gaussian component and the jump component:
\begin{equation*}
    X_t = X_0\mathrm{e}^{-t} + \int_0^t \mathrm{e}^{-(t-s)}dB_s + \int_0^t \mathrm{e}^{-(t-s)}dL^{\alpha}_s.
\end{equation*}
The above jump diffusion $X_t$ has generator
\begin{align*}
  \mathcal{L}\varphi(x)
  = & -x \nabla\varphi(x) + \frac{1}{2}\Delta \varphi(x) + \int_{\mathbb{R}\setminus\{0\}} [\varphi(x+z) - \varphi(x)]\frac{C_{1,\alpha}}{|z|^{1+\alpha}}dz \\
  = & -x \nabla\varphi(x) + \frac{1}{2}\Delta\varphi(x) + \int_{\mathbb{R} \setminus \{x\}} [\varphi(y) - \varphi(x)] k(y,x)dy \\
  = & -x \nabla\varphi(x) + \frac{1}{2}\Delta\varphi(x) + \Delta^{\alpha/2}\varphi(x),
\end{align*}
where $k(x,y) = C_{1,\alpha}|x-y|^{-(1+\alpha)}$ is the jump kernel associated with the fractional Laplacian operator $\Delta^{\alpha/2}$.

To find the density function, we first calculate the characteristic function (Fourier transform) of the jump part of $X_t$. For a stochastic integral of the form $\int f(s) dL^{\alpha}_s$,  its characteristic function can be calculated using the L\'{e}vy-Khintchine formula. Generally, for a deterministic function $f(s)$, we have:
\begin{equation*}
    \mathbb{E}\left[ \exp\left( i\xi\int_0^t f(s)dL^{\alpha}_s \right) \right] =  \exp\left( -\int_0^t |\xi f(s)|^{\alpha}ds \right)
\end{equation*}
due to the homogeneity of the L\'{e}vy measure and the symmetry of the process. Here $f(s) = \mathrm{e}^{-(t-s)}$, so
\begin{equation*}
    -\int_0^t |\xi f(s)|^{\alpha}ds = -|\xi|^\alpha\int_0^t \mathrm{e}^{-(t-s)}ds =- \frac{1}{\alpha}(1-\mathrm{e}^{-\alpha t})|\xi|^\alpha
\end{equation*}
is a L\'{e}vy symbol of stable L\'{e}vy processes with scale parameter
\begin{equation*}
    \sigma_t^\alpha = \frac{1}{\alpha}(1-\mathrm{e}^{-\alpha t}).
\end{equation*}
Therefore, the characteristic function is as follows
\begin{equation*}
    \Phi(\sigma_t\xi) = \mathbb{E}\left[ \exp\left( i\xi\int_0^t \mathrm{e}^{-(t-s)}dL^{\alpha}_s \right) \right] = \exp(-|\sigma_t\xi|^\alpha ).
\end{equation*}
And the density of the Gaussian part is a scaled Gaussian density with scale parameter $c^2_t=1-\exp(-2t)$
\begin{equation*}
    \rho(t,x) = \frac{1}{\sqrt{\pi}c_t}\exp\left( -\frac{x^2}{c^2_t} \right).
\end{equation*}
Therefore, the characteristic function of the solution $X_t$ is obtained by multiplying the characteristic functions of the two Gaussian distributions (from the initial value and the continuous part) and the characteristic function of the jump component's distribution. Then, applying the inversion formula, we have
\begin{equation*}
    \rho(t,x) =  \frac{1}{\pi}\int_0^\infty \cos((x-3\mathrm{e}^{-t})\xi) \exp \left( -\frac{1}{4}( 1 - \frac{4}{5}\mathrm{e}^{-2t} )\xi^2 -\frac{1}{\alpha}(1-\mathrm{e}^{-\alpha t})\xi^\alpha\right) d\xi,
\end{equation*}
and the invariant measure $\mu_\mathrm{ss}$ with density:
\begin{equation}
    \rho_\mathrm{ss}(x)  = \frac{1}{\pi}\int_0^\infty \cos(x\xi) \exp \left( - \frac{1}{4}\xi^2 -\frac{1}{\alpha}\xi^\alpha\right) d\xi .
\end{equation}

Using the density function above, we plot \Cref{fig:placeholderX}. As indicated by the exponential ergodicity of the equation, the distributions differ markedly in the early evolution stage ($t$ between 0 and 0.5), but become very close to the stationary distribution already at $t=2$.

\begin{figure}
    \centering
    \includegraphics[width=0.75\linewidth]{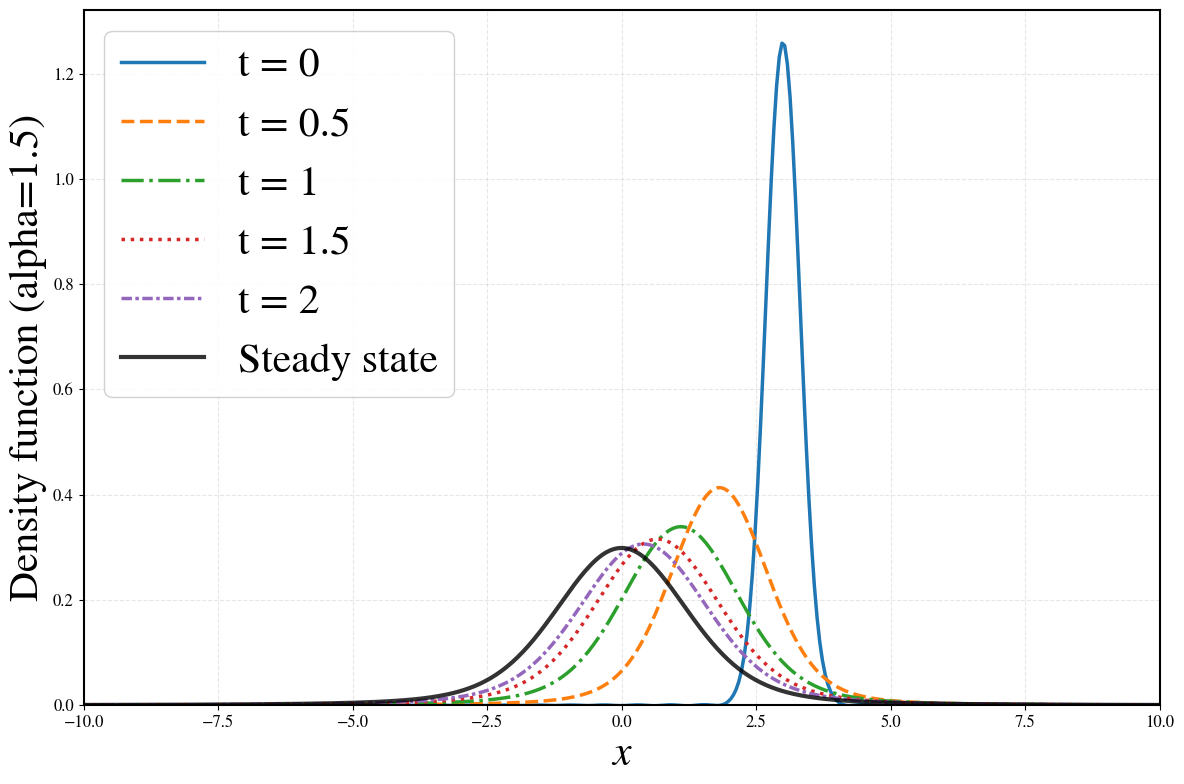}
    \caption{Time Evolution of SDE Density Function from $N(3,0.1)$. The initial condition is a concentrated normal distribution $N(3, 0.1)$ at $t=0$ (solid blue). As time progresses, the density broadens, flattens, and shifts, ultimately converging to its steady-state (black line), which exhibits the characteristic heavy-tailed profile of a Lévy-stable distribution.}
    \label{fig:placeholderX}
\end{figure}

We then plot the free energy \Cref{fig:Free energy} using the density data. This free energy, defined by the relative entropy, captures the energy dissipation driven by the discrepancy between the initial and stationary densities. Consequently, as the distribution converges exponentially to the steady state, the ratio $\rho(t,x)/\rho^\mathrm{ss}(x)$ quickly approaches $1$, and its logarithm (or after a log transform) decays to $0$. Since the free energy decays sharply near time $0$ and changes gently as it approaches equilibrium, the plot of the free energy dissipation rate is very steep at the initial position and quickly tends toward zero (see \Cref{fig:EPR-FED}). This is also consistent with the conclusion in Section \ref{Sec3}.

\begin{figure}[htbp]
    \centering
    \includegraphics[width=0.75\linewidth]{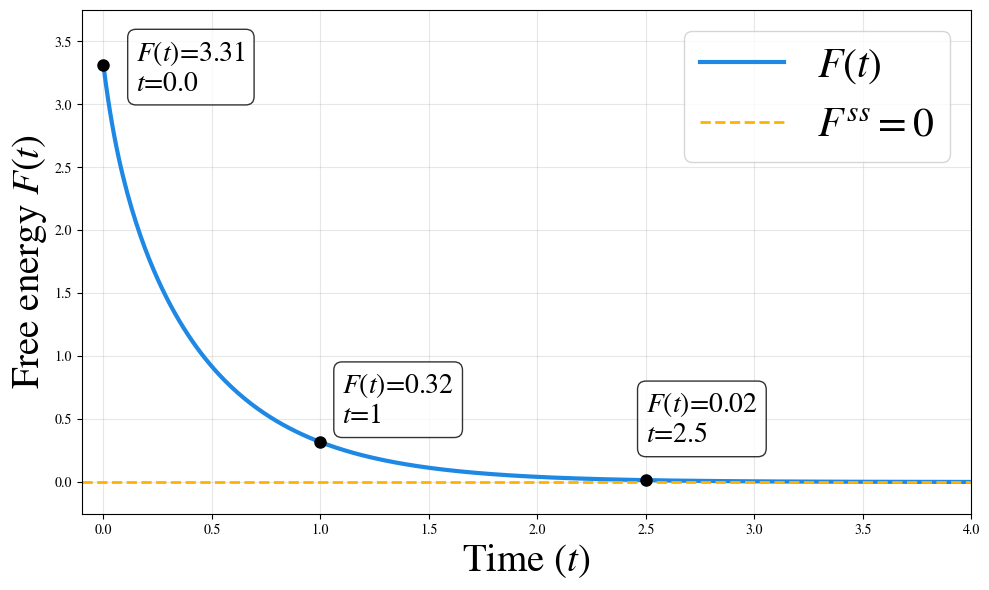}
    \caption{Temporal dissipation of free energy $F(t)$. At the initial time $t=0$, the free energy is approximately $F(0) = 3.31$. It then undergoes rapid dissipation: by $t=1$, more than 90\% of the initial free energy has dissipated, with $F(1) = 0.32$. At $t = 2.5$, about 99\% of the initial free energy has been lost, yielding $F(2.5) = 0.02$. Ultimately, the free energy approaches zero, indicating complete dissipation. The dashed orange line marks the steady-state free energy $F^{\mathrm{}{ss}} = 0$.}
\label{fig:Free energy}
\end{figure}

It is straightforward to verify that if both the density of $X_t$ and its steady-state density exist, then the steady state of $X_t$ fails to satisfy the detailed balance condition at the steady state, i.e., $j^\mathrm{loc}\neq0$ and $j^\mathrm{nl}\neq0$.
In physical intuitive understanding, because the confining drift $b(x)=-x$ concentrates probability mass near the origin, there are many more particles available to make large outward jumps than there are particles in the tails to make equally large inward jumps; the resulting net jump current is balanced by the restoring drift only at the level of stationarity, not by pairwise cancellation of forward/backward transitions. This persistent circulation (nonzero $j^{\mathrm{loc}}$ and $j^{\mathrm{nl}}$) is precisely the signature of broken detailed balance and of a non-equilibrium steady state.

Since the density function of $X_t$ is spatially differentiable, with the temperature simplified to unity, the free energy dissipation is
\begin{align*}
    \frac{dF}{dt}
    =& -\int_{\mathbb{R}} j^\mathrm{loc}(t,x)\nabla\left( \log\left(\frac{\rho_\mathrm{ss}(x)}{\rho(t,x)}\right) \right) dx \\
    &- \frac{1}{2}\int_{\mathbb{R}}\int_{\mathbb{R}^n \setminus \{x\}} j^\mathrm{nl}(t,x,y)\log \left( \frac{\rho(t,x) \rho_\mathrm{ss}(y)}{\rho(t,y)\rho_\mathrm{ss}(x)} \right) dydx\\
    =& \int_{\mathbb{R}} \left( x\rho(t,x) + \frac{1}{2}\nabla\rho(t,x) \right) \nabla\left( \log\left(\frac{\rho_\mathrm{ss}(x)}{\rho(t,x)}\right) \right)dx\\
    &- \int_{\mathbb{R}}\int_{\mathbb{R} \setminus \{x\}} \rho(t,x)k(x,y)\log \left( \frac{\rho(t,x) \rho_\mathrm{ss}(y)}{\rho(t,y)\rho_\mathrm{ss}(x)} \right) dydx.
\end{align*}
It can be decomposed into the housekeeping heat
\begin{align*}
    Q_\mathrm{hk}
    =& \int_{\mathbb{R}}\left( x\rho(t,x) + \frac{1}{2}\nabla\rho(t,x) \right)\left( 2x + \frac{\nabla\rho_\mathrm{ss}(x)}{\rho_\mathrm{ss}(x)}\right)dx\\
    &+ \frac{1}{2}\int_{\mathbb{R}}\int_{\mathbb{R} \setminus \{x\}} [\rho(t,x)k(x,y) - \rho(t,y)k(y,x)]\log \left( \frac{\rho_\mathrm{ss}(x)}{\rho_\mathrm{ss}(y)} \right) dydx
\end{align*}
and entropy production rate
\begin{align*}
    e_\mathrm{p} 
    =& \frac{1}{2}\int_\mathbb{R}\rho(t,x)\left( 2x + \frac{\nabla\rho(t,x)}{\rho(t,x)} \right)^2 dx\\
    &+ \frac{1}{2}\int_{\mathbb{R}}\int_{\mathbb{R} \setminus \{x\}} [\rho(t,x)k(x,y) - \rho(t,y)k(y,x)]\log \left( \frac{\rho(t,x)}{\rho(t,y)} \right) dydx.
\end{align*}
From the expression of housekeeping heat, one can read off the associated local and nonlocal thermodynamic forces as $f^\mathrm{loc} = -2x-(\nabla\rho_\mathrm{ss})/\rho_\mathrm{ss}$ and $f^\mathrm{nl}=\log\rho_\mathrm{ss}(x)-\log\rho_\mathrm{ss}(y)$.
\begin{figure}[htbp]
    \centering
    \begin{subfigure}[b]{0.45\textwidth}
        \centering
        \includegraphics[width=\textwidth]{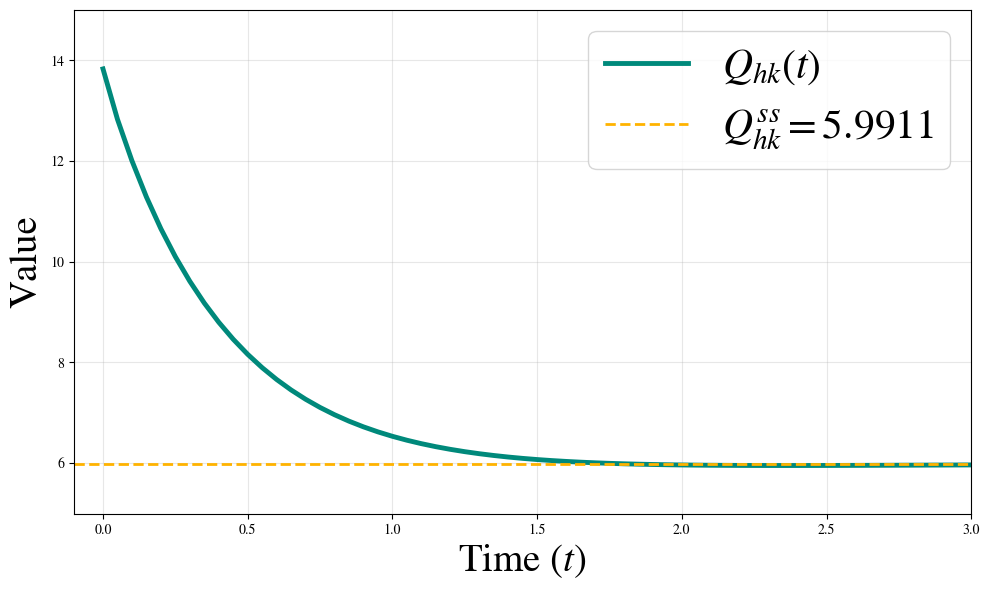}
        \caption{}
        \label{fig:HKH}
    \end{subfigure}
    \hfill
    \begin{subfigure}[b]{0.45\textwidth}
        \centering
        \includegraphics[width=\textwidth]{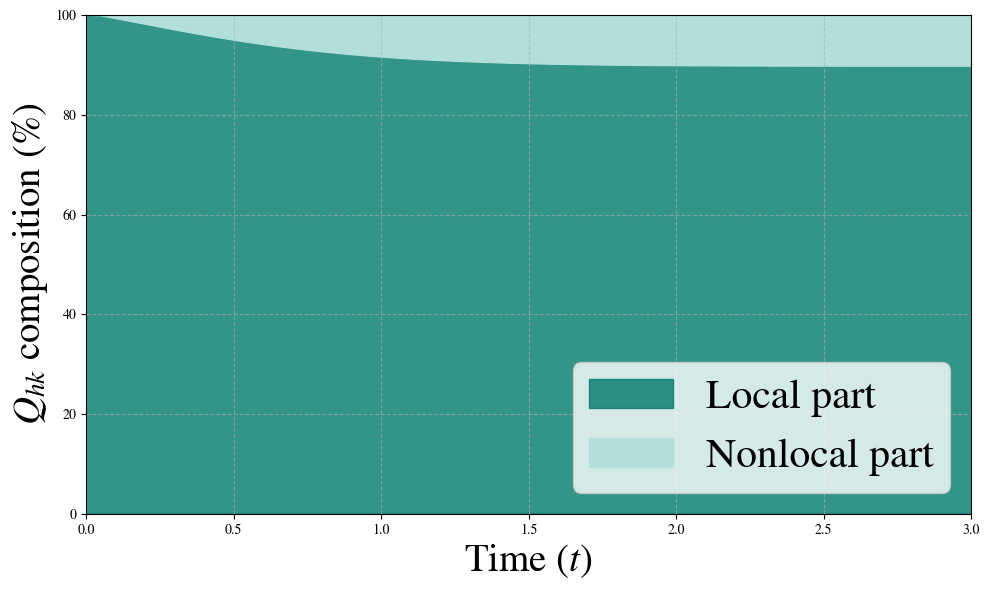}
        \caption{}
        \label{fig:Qdem}
    \end{subfigure}
    
    \vspace{0.5cm}
    \begin{subfigure}[b]{0.45\textwidth}
        \centering
        \includegraphics[width=\textwidth]{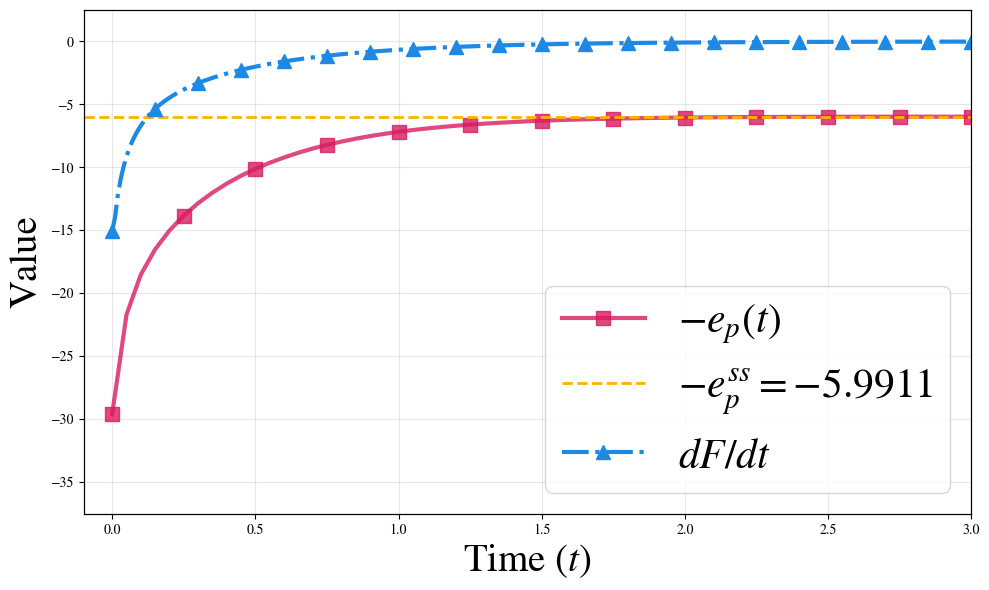}
        \caption{}
        \label{fig:EPR-FED}
    \end{subfigure}
    \hfill
    \begin{subfigure}[b]{0.45\textwidth}
        \centering
        \includegraphics[width=\textwidth]{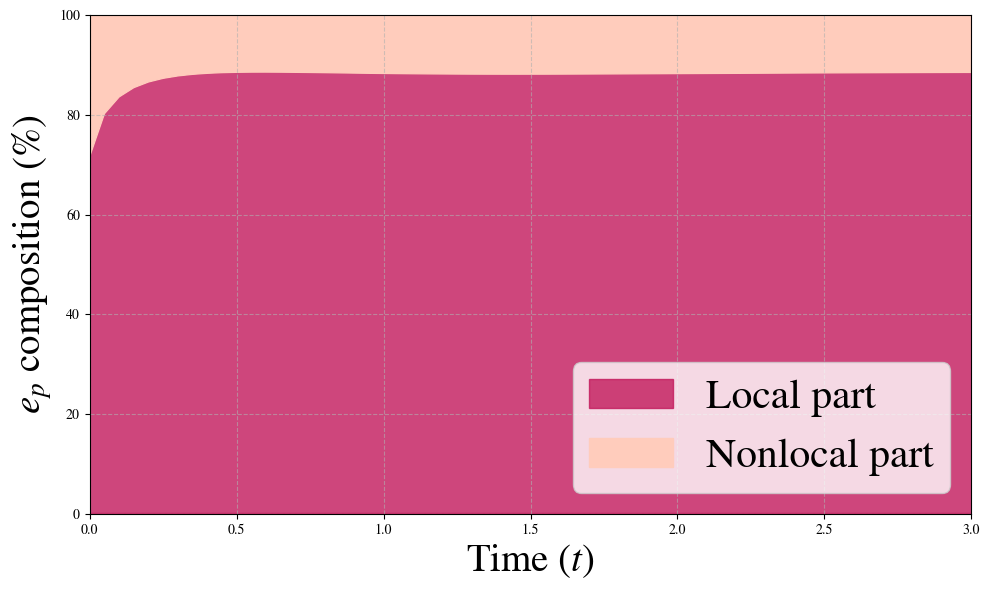}
        \caption{}
        \label{fig:Edem}
    \end{subfigure}
    
    \vspace{0.5cm}
    \begin{subfigure}[b]{0.45\textwidth}
        \centering
        \includegraphics[width=\textwidth]{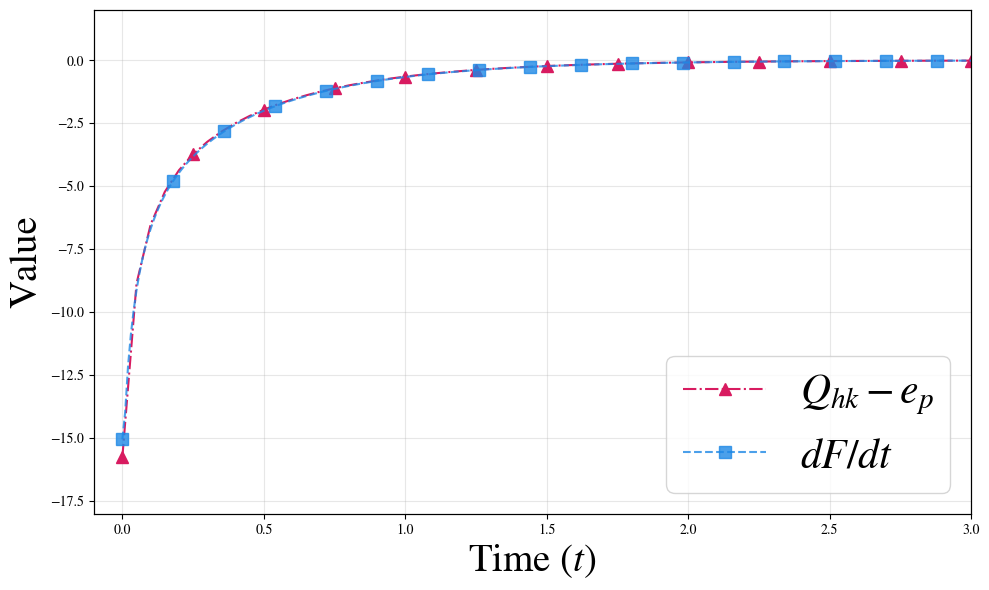}
        \caption{}
        \label{fig:FFF}
    \end{subfigure}
    \hfill
    \begin{subfigure}[b]{0.45\textwidth}
        \centering
        \includegraphics[width=\textwidth]{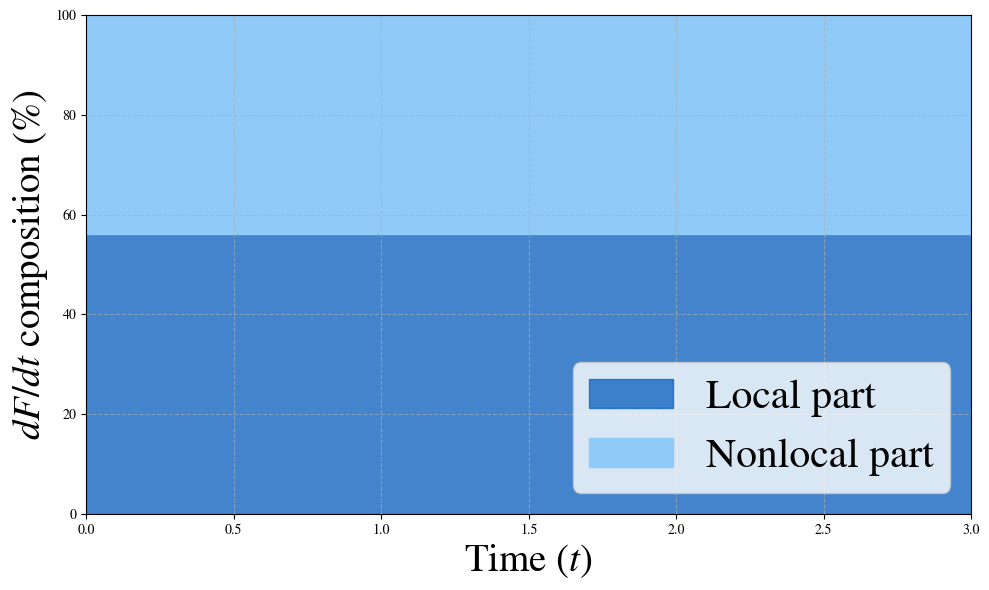}
        \caption{}
        \label{fig:Fdem}
    \end{subfigure}
    \caption{Time evolution of key thermodynamic rates and their local-nonlocal composition. Left column (panels \textbf{a}, \textbf{c} and \textbf{e}) displays the time-dependent behavior, and the right column (panels \textbf{b}, \textbf{d} and \textbf{f}) shows the corresponding percentage composition of local and nonlocal contributions for each quantity. The housekeeping heat $Q_\mathrm{hk}(t)$ (blue solid curve in panel \textbf{a}) and the (negative) entropy production rate $-e_\mathrm{p}(t)$ (magenta squares in panel \textbf{c}) both decay from their initial values and converge toward distinct, non-zero steady states $Q^\mathrm{ss}_\mathrm{hk}=e^\mathrm{ss}_\mathrm{p}(t)=5.9911$ (golden yellow line in panels \textbf{a} and \textbf{c}). In panels \textbf{b} and \textbf{d}, the local part (dark green for $Q_\mathrm{hk}(t)$, deep magenta for $-e_\mathrm{p}(t)$) dominates the respective quantity throughout the relaxation. In contrast, for the free‑energy dissipation rate shown in panel \textbf{f}, the local (blue) and nonlocal (light blue) contributions become comparable in magnitude, reflecting the balanced local/nonlocal composition of $dF/dt$.}
    \label{fig:valueofX}
\end{figure}

In \eqref{EX1SDE}, the drift $b(x)=-\nabla V(x)=-x$ represents an effective mean-reverting confinement (tethering, and crowding-induced restoring tendencies). 
The Gaussian noise arises from the small displacements generated by collisions between the particle and other particles in the cytoplasm—often referred to as diffusion—and also includes directed diffusion such as concentration gradients actively maintained by the cell. 
The $\alpha$-stable term represents intermittent burst-like relocations at the observational scale; it should be viewed as an idealized limit capturing heavy-tailed displacement statistics rather than a literal instantaneous microscopic trajectory.

As shown in \Cref{fig:EPR-FED}, the entropy production rate of the system does not approach zero. This occurs because the external thermodynamic force $f^\mathrm{loc}$ and $f^\mathrm{nl}$ continuously performs work on the system; \Cref{fig:HKH} shows that the power of the external thermodynamic force $dW/dt=Q_\mathrm{hk}$ also does not converge to zero. This precisely corresponds to the external energy supply inherent to active transport processes. A comparison between \Cref{fig:EPR-FED} and \Cref{fig:FFF} confirms $\frac{dF}{dt}=Q_\mathrm{hk}-\beta^{-1}e_\mathrm{p}$: while $e_\mathrm{p}$ and $dF/dt$ fail to coincide in \Cref{fig:EPR-FED}, the introduction of $Q_\mathrm{hk}$ in \Cref{fig:HKH} and \Cref{fig:FFF} reveals the complete decomposition of $dF/dt$ (the plot for $dF/dt$, obtained via a finite-difference scheme, exhibits a small residual error at the origin). This behavior stands in clear contrast to systems subject to no external force (see \eqref{EXSDE2} or \cite{ZL25}), where the free energy dissipation rate would be fully accounted for by the entropy production rate.

Under our parameter settings, \Cref{fig:Qdem} and \Cref{fig:Edem} indicate that the local parts of $e_\mathrm{p}(t)$ and $Q_\mathrm{hk}(t)$ are dominant. In contrast, for the free-energy dissipation rate shown in panel \Cref{fig:Fdem}, the local and nonlocal contributions become comparable in magnitude. Variations in the jump intensity $\lambda$ of the noise (see Section \ref{Sec2}) or in the equation coefficients ($b$, $a$ and $\sigma$) will alter the scales of these components relative to each other. The specific configuration, however, is a data-dependent choice informed by actual experimental results.

We decompose the generator of \eqref{EX1SDE} as $\mathcal{L}=\mathcal{L}_\mathrm{s}+\mathcal{L}_\mathrm{a}$, where
\begin{align*}
    \mathcal{L}_\mathrm{s}\varphi(x)
    &= \frac{1}{2}\nabla\varphi(x)\frac{\nabla\rho_\mathrm{ss}(x)}{\rho_\mathrm{ss}(x)} + \frac{1}{2}\Delta\varphi(x) + \frac{1}{2}\int_{\mathbb{R} \setminus \{x\}}(\varphi(y)-\varphi(x))\left( 1 + \frac{\rho_\mathrm{ss}(y)}{\rho_\mathrm{ss}(x)} \right)k(x,y)dy,\\
    \mathcal{L}_\mathrm{a}\varphi(x)
    &= \left( -x - \frac{1}{2}\frac{\nabla\rho_\mathrm{ss}(x)}{\rho_\mathrm{ss}(x)} \right)\nabla\varphi(x) + \frac{1}{2}\int_{\mathbb{R} \setminus \{x\}}(\varphi(y)-\varphi(x))\left( 1 - \frac{\rho_\mathrm{ss}(y)}{\rho_\mathrm{ss}(x)} \right)k(x,y)dy.
\end{align*}
Thus, an alternative decomposition of the free energy dissipation is achieved through the decomposition of the generator, with
\begin{align*}
    F'_\mathrm{s}(\rho(t,x))
    =& -\frac{1}{2}\int_{\mathbb{R}}\rho(t,x)\left(\frac{\nabla\rho_\mathrm{ss}(x)}{\rho_\mathrm{ss}(x)} - \frac{\nabla\rho(t,x)}{\rho(t.x)}\right)^2dx\\
    &- \frac{1}{2}\int_{\mathbb{R}}\int_{\mathbb{R} \setminus \{x\}} k(x,y)\rho(t,x)\left( 1 + \frac{\rho_\mathrm{ss}(y)}{\rho_\mathrm{ss}(x)} \right)\log\left(\frac{\rho(t,x)\rho_\mathrm{ss}(y)}{\rho(t,y)\rho_\mathrm{ss}(x)}\right) dydx,\\
    F'_\mathrm{a}(\rho(t,x))
    =& 0.
\end{align*}
$\mathcal{L}=\mathcal{L}_\mathrm{s}+\mathcal{L}_\mathrm{a}$ represents the decomposition of the operator relative to the invariant measure $\mu_\mathrm{ss}$ of \eqref{EX1SDE}. The symmetric part is self-adjoint and generates reversible (detailed-balance) dynamics, governing the monotonic decay of free energy; the anti-symmetric part produces a measure-preserving conservative circulation, and therefore its contribution to free energy dissipation is identically zero, i.e., $F'_\mathrm{a}(\rho(t,x)) = 0$.

\Cref{fig:placeholderX} - \Cref{fig:valueofX} visualize the identity $\frac{dF}{dt}=Q_\mathrm{hk}-\beta^{-1}e_\mathrm{p}$ and the splitting $\mathcal{L}=\mathcal{L}_s+\mathcal{L}_a$: relaxation ($F\downarrow$) is governed by the reversible sector $\mathcal{L}_\mathrm{s}$, while persistent non-equilibrium maintenance (plateaus of $Q_\mathrm{hk}$ and $e_\mathrm{p}$) is encoded by the conservative sector $\mathcal{L}_\mathrm{a}$.

\subsection{Intracellular inactive transport with equilibrium steady state}

Now we consider an ``inactivation-like'' scenario in which sustained non-equilibrium driving is suppressed and the system relaxes toward an equilibrium steady state satisfying detailed balance. 
In such a case, the housekeeping heat vanishes at steady state and the entropy production decays to zero as the distribution is equilibrium.
The most fundamental distinction between the steady state of the inactivation-like process (the inactivated-like state) and the quiescent-like state is that the steady state of the quiescent-like state remains non-equilibrium.

We choose an explicitly solvable ``reset-to-Gibbs'' jump mechanism to obtain closed-form densities and directly verify $Q_\mathrm{hk}\equiv 0$ and $dF/dt=-\beta^{-1}e_\mathrm{p}$. 
Consider the jump diffusion
\begin{equation}\label{EXSDE2}
    dY_t = -Y_tdt + \sqrt{2}dB_t + \int_{\mathbb{R} \setminus \{0\}}(z-Y_t) N(dt,dz),
\end{equation}
where $B_t$ is a one-dimensional standard Brownian motion and $N_t$ is the Poisson random measure $dt\nu(dz) = \exp(-\frac{1}{2}z^2)dzdt$. 

The explicit expression for the density function of this process can be intuitively explained by its dynamics: at any time $t$, the distribution of the state $Y_t$ is a mixture distribution. If no jump occurs in the interval $[0, t]$ (with probability $\mathrm{e}^{-\sqrt{2\pi}t}$), then $Y_t$ follows the distribution of an Ornstein-Uhlenbeck process; if at least one jump occurs, then due to the reset property of the jumps, the distribution of $Y_t$ is exactly the stationary distribution (the Gibbs measure). Therefore, let $Y_0=0$, we can use the combination of the jump part and the Gaussian part as its characteristic function, and directly obtain its density function using the inversion formula
\begin{equation*}
    \rho(t,y) = \frac{1}{\pi}\int_0^\infty cos(y\xi)\left[ \mathrm{e}^{-\sqrt{2\pi}t}\exp\left( -\frac{1}{2}(1-\mathrm{e}^{-2t})\xi^2 \right) + (1 - \mathrm{e}^{-\sqrt{2\pi}t})\exp\left( -\frac{1}{2}\xi^2 \right) \right] d\xi.
\end{equation*}
And the invariant measure of the jump process $Y_t$ is a Gibbs distribution $\mu(dy) = \frac{1}{\sqrt{2\pi}}\exp(-\frac{1}{2}y^2)dy$. The evolution plot of the density function \Cref{fig:placeholderY} shows that the distribution is already very close to the steady state at about $t=0.6$, even though it started from a Dirac delta function. As the Dirac delta function introduces a generalized function, the red curve serves only as a schematic illustration.

\begin{figure}
    \centering
    \includegraphics[width=0.75\linewidth]{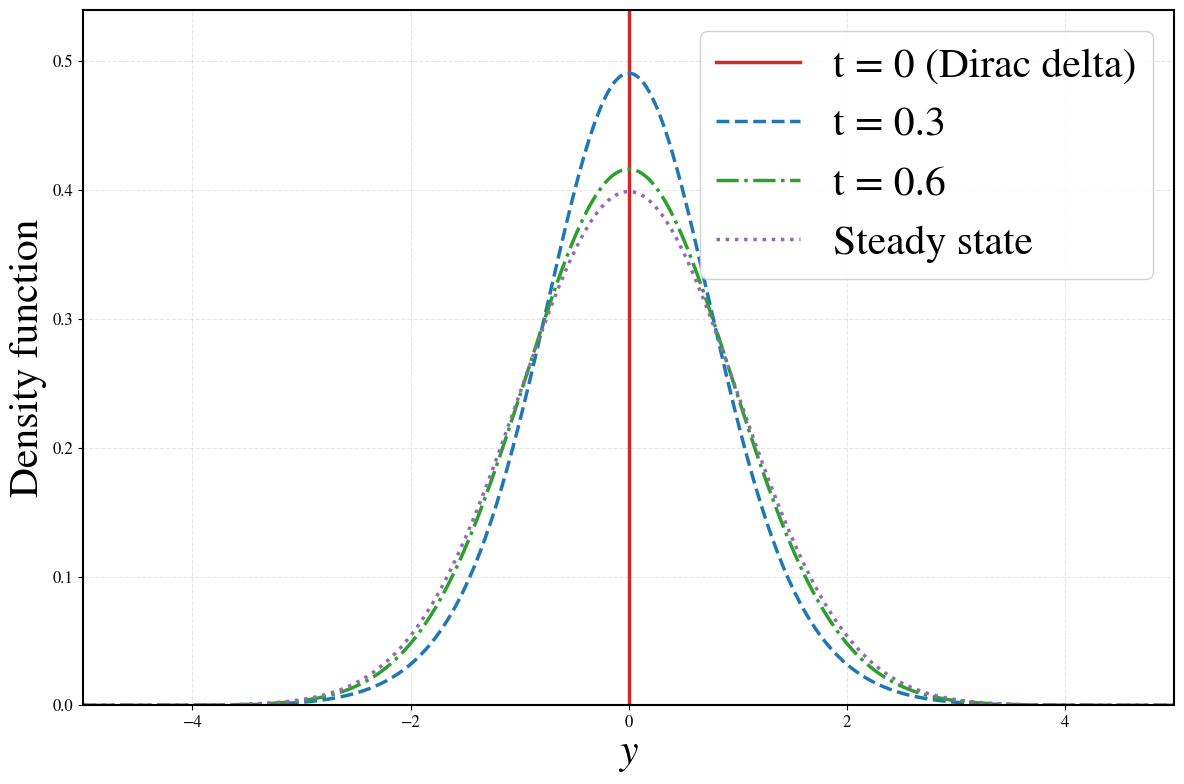}
    \caption{Time Evolution of SDE Density Function from $0$. The density $\rho(y,t)$ is shown at successive times: $t=0$ (red solid line), $t=0.3$ (blue dashed line), and $t=0.6$ (green dash-dotted line), converging toward the steady-state distribution (purple dotted line). The initial Dirac delta function is represented by a sharply localized profile. The solution demonstrates the progressive broadening and smoothing of the density from a concentrated initial state to a diffuse equilibrium distribution.}
    \label{fig:placeholderY}
\end{figure}

The entropy production rate of system $Y_t$ is given by
\begin{equation*}
    e_\mathrm{p} = \int_\mathbb{R}\rho(t,y)\left( x + \frac{\nabla\rho(t,y)}{\rho(t,y)} \right)^2 dy + \int_{\mathbb{R}}\int_{\mathbb{R} \setminus \{y\}} \rho(t,y)\mathrm{e}^{-\frac{1}{2}z^2}\left[ \log \left( \frac{\rho(t,y)}{\rho(t,z)} \right) + \frac{1}{2}(y^2-z^2)\right] dzdy.
\end{equation*}
It can be verified that the steady state of $Y_t$ is an equilibrium state, therefore the free energy dissipation rate of $Y_t$ is entirely given by its entropy production rate \cite{ZL25}.
\begin{figure}
    \centering
    \includegraphics[width=0.75\linewidth]{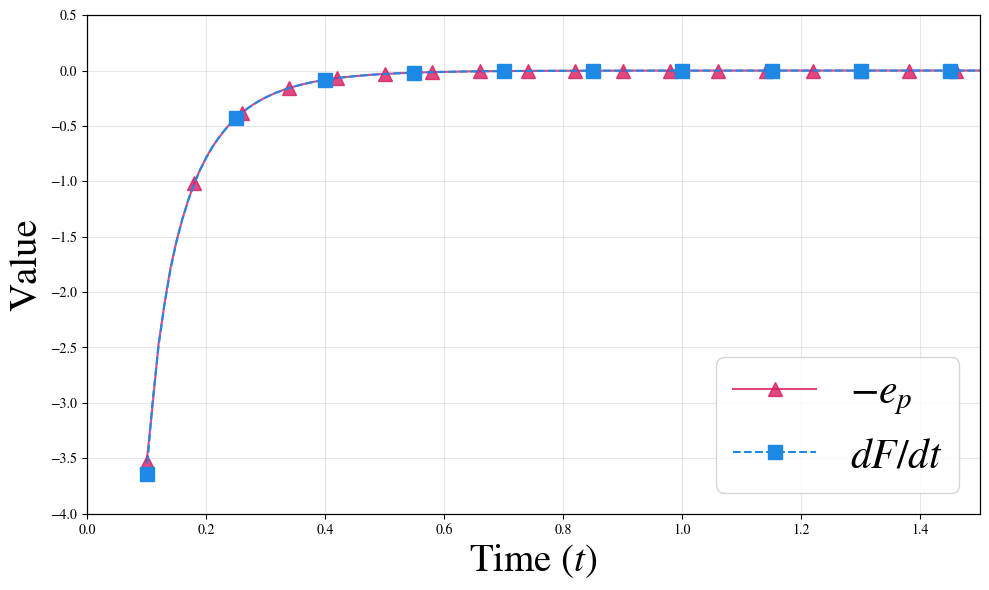}
    \caption{Transient Convergence of Free Energy Dissipation and Entropy Production Rates. The blue squares represent the time derivative of the free energy $dF/dt$, and the red triangles represent the negative entropy production rate, $-e_\mathrm{p}$. Due to the singular nature of the Kullback-Leibler divergence for the initial Dirac delta distribution, calculations commence from $t = 0.1$. The near-perfect overlap of the two curves confirms the theoretical equality $dF/dt = -e_\mathrm{p}$ throughout the transient evolution until both quantities approach zero at the steady state.}
    \label{fig:FFFY}
\end{figure}

\Cref{fig:FFFY} shows the comparison between the negative value of the entropy production rate and the free energy dissipation rate of system $Y_t$, where the two curves are seen to overlap almost perfectly. The comparison between \Cref{fig:valueofX} and \Cref{fig:FFFY} also clearly demonstrates that sustained non-equilibrium driving (e.g., active fluctuations and steady probability currents) is responsible for the persistent positive entropy production at stationarity.

Example \eqref{EXSDE2} provides a minimal prototype of an inactivation-like regime: the jump mechanism acts as a rapid reset toward the Gibbs measure, and the resulting stationary state satisfies detailed balance. Accordingly, the non-equilibrium maintenance cost vanishes, $Q_\mathrm{hk}(t)=0$, and the free-energy dissipation is entirely accounted for by the entropy production, $dF/dt=-\beta^{-1}e_\mathrm{p}(t)$, consistent with \Cref{fig:FFFY}. From the viewpoint of our operator splitting $\mathcal{L}=\mathcal{L}_\mathrm{s}+\mathcal{L}_\mathrm{a}$, this corresponds to a situation where the antisymmetric sector is absent (or negligible), so that no stationary circulation persists once the steady state is reached. In contrast, an activation-like regime can be viewed (at the same coarse-grained level) as one with stronger sustained driving and/or more pronounced jump-mediated transport, for which $\mathcal{L}_\mathrm{a}$ generates non-vanishing stationary currents. In such a regime, $F(t)$ may still relax to zero, yet $Q^\mathrm{ss}_\mathrm{hk}(t) = \beta^{-1}e^\mathrm{ss}_\mathrm{p}(t) > 0$ remains as a persistent energetic cost of maintaining non-equilibrium organization; empirically this often manifests as larger steady dissipation plateaus and sharper transient rates.

\section{Conclusion}

In this work, we have established a structural decomposition for jump--diffusion processes by splitting the generator into symmetric and anti-symmetric components in the weighted space $L^2(\rho_\mathrm{ss})$. 
This decomposition provides a unified thermodynamic interpretation of nonequilibrium dynamics: the symmetric part corresponds to a reversible mechanism responsible for free energy dissipation, while the anti-symmetric part generates a canonical conservative flow sustaining probability currents without contributing to entropy decay.

From the perspective of non-equilibrium dynamics, this framework naturally clarifies the distinction between detailed balance and broken detailed balance. 
When the anti-symmetric component vanishes, the dynamics satisfies detailed balance and relaxes toward equilibrium purely through gradient flow driven by free energy. 
In contrast, a nontrivial anti-symmetric component breaks detailed balance and induces persistent probability currents in the steady state, leading to non-equilibrium stationary states (NESS). 
In this regime, entropy production admits a decomposition into a dissipative contribution associated with the symmetric generator and a housekeeping part that maintains irreversible circulation.

Our results also admit a natural interpretation in terms of the currents--forces structure and the Onsager viewpoint. 
The generalized Fisher information associated with the symmetric generator plays the role of a quadratic dissipation functional, while the corresponding probability currents can be expressed as fluxes driven by thermodynamic forces (gradients of the log-density). 
Near equilibrium, this leads to a linear response regime consistent with Onsager reciprocity relations. 
Beyond equilibrium, the decomposition highlights how non-gradient (rotational) components of the dynamics generate circulating currents that are orthogonal to the thermodynamic forces and therefore do not contribute to free energy dissipation.

A key outcome of this work is the identification of a generalized gradient flow structure for jump--diffusion processes. 
While classical diffusion corresponds to a Wasserstein gradient flow of free energy, the present framework extends this picture to nonlocal dynamics, where the dissipation is governed by a combination of local and jump-induced Fisher information. 
The anti-symmetric component, in contrast, can be viewed as a Hamiltonian-like transport on the space of probability measures, leading to conservative motion along level sets of the free energy.

These structural insights are directly relevant for stochastic dynamics in biological and chemical systems. 
In biochemical reaction networks, molecular motors, and intracellular transport, stochastic dynamics often operate far from equilibrium and exhibit sustained probability currents driven by chemical energy input. 
Such systems can be modeled by jump processes or jump--diffusions, where the anti-symmetric component captures active driving forces (e.g.\ ATP consumption), while the symmetric part reflects thermal fluctuations. 
Similarly, in chemical reaction systems described by master equations, entropy production and housekeeping heat quantify the energetic cost of maintaining non-equilibrium steady states. 
The present decomposition provides a systematic way to separate dissipative relaxation from driven circulation in these contexts.

There are several problems left to be considered in the future:

\medskip
\noindent
\textbf{(i) Nonlocal gradient flow structures.}
While in the purely diffusive setting this structure is rigorously understood as a Wasserstein gradient flow \cite{AGS05}, a systematic and rigorous characterization of the underlying metric structure associated with jump–diffusions, nonlocal gradient flow, and generalized Wasserstein distances remains an open problem.

In particular, it is natural to ask whether the nonlocal Fisher information can be associated with a suitable metric or action functional, extending the Benamou--Brenier formulation of optimal transport to jump processes. 
Recent developments \cite{E2014, PR22} in nonlocal transport metrics and Lévy-driven dynamics indicate that the jump contribution may correspond to a transport mechanism involving discontinuous paths and long-range interactions. 
In our framework, the Dirichlet form associated with $\mathcal L_s$ can define a generalized mobility operator, while the logarithmic derivative $\nabla \log \rho$ (or its nonlocal analogue) plays the role of a thermodynamic force.
Establishing the associated contraction properties, displacement convexity of free energy, and corresponding functional inequalities in such nonlocal settings would provide a deeper geometric understanding of jump-driven systems.

\medskip
\noindent
\textbf{(ii) Entropy production and fluctuation relations.}
The decomposition of the generator naturally leads to a refined structure of entropy production in non-equilibrium systems. 
It would be of interest to relate the present decomposition and total entropy production rate to fluctuation theorems, particularly for jump-driven systems far from equilibrium.
In the spirit of macroscopic fluctuation theory \cite{esposito2010,stochasticThermo1}, one expects that the entropy production rate appears as a quadratic functional governing fluctuations of probability currents around their typical values. 
Extending such large deviation structures to jump--diffusion systems, and identifying the role of the symmetric and anti-symmetric components in the corresponding rate functionals, remains an important open problem.
Understanding how the decomposition interacts with fluctuation theorems may provide a deeper thermodynamic interpretation of irreversible jump processes, particularly in regimes far from equilibrium.

\medskip
\noindent
\textbf{(iii) Potential applications in biological systems and machine learning.}
An important direction is to apply the present decomposition to jump--diffusion dynamics from active biological systems, such as molecular motors, gene regulatory networks, and cellular signaling pathways \cite{briane2020,chen2015}. 
Within our framework, the symmetric component captures passive, thermally induced fluctuations and governs free energy dissipation, while the anti-symmetric component encodes active driving forces that maintain non-equilibrium steady states. 
This separation provides a quantitative way to distinguish dissipative relaxation from active energy consumption, and may offer new tools to study efficiency, robustness, and information processing in living systems.

At the same time, jump--diffusion processes are increasingly relevant in modern machine learning, including high-dimensional stochastic optimization, and generative modeling with non-Gaussian noise \cite{BV24,HLMZ25,HHY2025}. 
The present thermodynamic decomposition suggests a principled way to incorporate non-reversible dynamics to accelerate convergence or enhance sampling efficiency. 
Moreover, the interpretation of free energy and Fisher information may provide new insights into training dynamics, regularization, and the role of noise in high-dimensional learning systems.

Taken together, these observations suggest that the interplay between gradient flow (dissipation) and conservative flow (circulation) may serve as a common organizing principle across biological and artificial systems. 
Understanding how to optimally balance these two mechanisms---minimizing entropy production while maintaining efficient exploration---is an intriguing direction for future research, with potential implications for both the thermodynamics of living systems and the design of scalable learning algorithms.

\medskip
\noindent
Overall, our results highlight that jump--diffusion systems provide a natural setting in which gradient flow, entropy production, and probability currents coexist and interact. 
We expect that further development of this perspective will contribute to a deeper understanding of irreversible processes in complex systems across physics, chemistry, and biology.

\section*{Acknowledgments}
The authors thank Prof. Jingqiao Duan, Dr. Qiao Huang, Dr. Yubin Lu, and Dr. Bin Miao, for their fruitful discussions on thermodynamics and stochastic processes.
QZ acknowledge support from the China Postdoctoral Science Foundation (Grant No. 2023M740331).
SYF acknowledge support from the NSFC grants 12141107 and W2541005, the Guangdong Provincial Key Laboratory of Mathematical and Neural Dynamical Systems (Grant 2024B1212010004), the Cross Disciplinary Research Team on Data Science and Intelligent Medicine (2023KCXTD054), and the Guangdong-Dongguan Joint Research Fund (Grant 2023A151514 0016).

\end{document}